\documentclass[10pt,twocolumn]{IEEEtran}

\IEEEoverridecommandlockouts

\usepackage{amsfonts,amsmath,amssymb,amsthm}
\usepackage{bbm,bm}
\usepackage{graphicx,epsfig,color}
\usepackage{algorithm,algorithmic}
\usepackage{comment}
\usepackage{epstopdf}
\usepackage{subfigure}

\newtheorem{thm}{Theorem}

\newtheorem{lem}{Lemma}[section]

\newtheorem{rem}{Remark}

\newcommand{\eg}{{\it e.g., }}

\newcommand{\ie}{{\it i.e., }}
\newcommand{\E}{{\mathbb E }}

\newcommand{\yv}{{\bf y}}

\newcommand*{\Scale}[2][4]{\scalebox{#1}{$#2$}}%

\begin{document}

\title{Optimal Dynamic Cloud Network Control}

\author{
\IEEEauthorblockN{
Hao Feng\IEEEauthorrefmark{1}\IEEEauthorrefmark{2}, Jaime Llorca\IEEEauthorrefmark{2}, Antonia M. Tulino\IEEEauthorrefmark{2}\IEEEauthorrefmark{3}, Andreas F.  Molisch\IEEEauthorrefmark{1}
}\\
\IEEEauthorblockA{
\IEEEauthorrefmark{1}University of Southern California, Email: \{haofeng, molisch\}@usc.edu\\
\IEEEauthorrefmark{2}Nokia Bell Labs, Email: \{jaime.llorca, a.tulino\}@nokia-bell-labs.com\\
\IEEEauthorrefmark{3}University of Naples Federico II, Italy. Email: \{antoniamaria.tulino\}@unina.it
}
\vspace{-0.15cm}}

\maketitle

\begin{abstract}
Distributed cloud networking enables the deployment of a wide range of services in the form of interconnected software functions instantiated over general purpose hardware at multiple cloud locations distributed throughout the network.
We consider the problem of optimal service delivery over a distributed cloud network, in which nodes are equipped with both communication and computation resources.
We address the design of distributed online solutions that drive flow processing and routing decisions, along with the associated allocation of cloud and network resources. 
For a given set of services, each described by a chain of service functions, we characterize the cloud network capacity region and design a family of dynamic cloud network control (DCNC) algorithms that
stabilize the underlying queuing system, while achieving arbitrarily close to minimum cost with a tradeoff in network delay.
The proposed DCNC algorithms make local decisions based on the online minimization of linear and quadratic metrics obtained from an upper bound on the Lyapunov drift-plus-penalty of the cloud network queuing system.
Minimizing a quadratic vs. a linear metric 
is shown to
improve the cost-delay tradeoff at the expense of increased computational complexity.
Our algorithms are further enhanced with
a shortest transmission-plus-processing distance bias that 
improves delay performance without compromising throughput or overall cloud network cost.
We provide throughput and cost optimality guarantees, convergence time analysis, and
extensive simulations in representative cloud network scenarios.
\end{abstract}

\section{Introduction}
\vspace{-0.25cm}
{\let\thefootnote\relax\footnote{Partial results have been presented in conference papers \cite{Infocom_workshop}, \cite{icc_2016}. This work
has been supported by NSF grant \# $1619129$ and CCF grant \# $1423140$.}}

Distributed cloud networking builds on network functions virtualization (NFV) and software defined networking (SDN)
to enable the deployment of network services in the form of elastic virtual network functions instantiated over general purpose servers at multiple cloud locations and interconnected via a programmable network fabric \cite{nfv_sdn,fxbook}. In
this evolved virtualized environment, \emph{cloud network} operators can host a variety of highly adaptable services over a common
physical infrastructure, reducing both capital and operational expenses, while providing quality of service guarantees.

Together with the evident opportunities of this attractive scenario, there come a number of technical challenges. Critical among them is
deciding where to execute each network function among the
various servers in the network. The ample opportunities for running network functions at multiple locations opens an interesting and challenging space for optimization. In addition, placement decisions must be coordinated with routing decisions that steer the network flows to the appropriate network functions, and with resource allocation decisions that determine the amount of resources (\eg virtual machines) allocated to each function.

The problem of placing virtual network functions in distributed cloud networks was first addressed in \cite{vnf}.
The authors formulate the problem as a generalization of facility location and generalized assignment,
and provide algorithms with bi-criteria approximation guarantees.
However, the model in \cite{vnf} does not capture \emph{service chaining}, where flows are required to go through a given sequence of service functions, nor flow routing optimization.
Subsequently, the work in [5] introduced a flow based model that allows optimizing the distribution (function placement and flow routing) of services with arbitrary function relationships (\eg chaining) over capacitated cloud networks. Cloud services are described via a directed acyclic graph and the function placement and flow routing is determined by solving a minimum cost network flow problem with service chaining constraints.

These studies, however, focus on the design of centralized solutions that assume global knowledge of service demands and network conditions.
With the increasing scale, heterogeneity, and dynamics inherent to both service demands and the underlying cloud network system, we argue that proactive centralized solutions 
must be complemented with distributed online algorithms that enable rapid adaptation to changes in network conditions and service demands, while providing global system objective guarantees.


In this work, we address the service distribution problem in a dynamic cloud network setting, where service demands are unknown and time-varying.
We provide the first characterization of a cloud network's capacity region and design 
throughput-optimal \emph{dynamic cloud network control} (DCNC) algorithms that drive local transmission, processing, and resource allocation decisions with global performance guarantees.
The proposed 
algorithms are based on applying the \emph{Lyapunov drift-plus-penalty} (LDP) control methodology \cite{Neely_book}-\cite{Neely_book2} to a cloud network queuing system that 
captures both the transmission and processing of service flows, consuming network and cloud resources.
We first propose DCNC-L, a control algorithm based on the minimization of a linear metric extracted from an upper bound of a quadratic LDP function of the underlying queuing system.
DCNC-L is a distributed joint flow scheduling and resource allocation algorithm that guarantees overall cloud network stability,
while achieving arbitrarily close to the minimum average cost with a tradeoff in network delay.
We then design DCNC-Q, an extension of DCNC-L that uses a quadratic metric derived from the same upper bound expression of the LDP function.
DCNC-Q preserves the throughput optimality of DCNC-L, and can significantly improve the cost-delay tradeoff at the expense of increased computational complexity.
Finally, we show that network delay can be further reduced by introducing a \emph{shortest transmission-plus-processing distance (STPD) bias} into the optimization metric.
The generalizations of DCNC-L and DCNC-Q obtained by introducing the shortest STPD bias are referred to as EDCNC-L and EDCNC-Q, respectively.


Our contributions can be summarized as follows:
\begin{itemize}
\item We introduce a queuing system for a general class of \emph{multi-commodity-chain (MCC)} flow problems that include the distribution of network services over cloud networks.
In our MCC queuing model, the queue backlog of a given commodity builds up, not only from receiving packets of the same commodity, but \emph{also} from processing packets of the preceding commodity in the service chain. 

\item For a given set of services, we characterize the capacity region of a cloud network in terms of the set of exogenous input flow rates that can be processed by the required service functions and
delivered to the required destinations, while maintaining the overall cloud network stability.
    Importantly, the cloud network capacity region depends on both the cloud network topology and the service structure. 

\item We design a family of throughput-optimal DCNC algorithms that {\em jointly schedule computation and communication resources} for flow processing and transmission without knowledge of service demands. The proposed algorithms allow pushing total resource cost arbitrarily close to minimum with 
a $[O(\epsilon),O(1/\epsilon)]$ cost-delay tradeoff, 
and they converge to within $O(\epsilon)$ deviation from the optimal solution in time $O(1/\epsilon^2)$.

\item Our DCNC algorithms make local decisions via the online minimization of 
linear and quadratic metrics extracted from
an upper bound of the cloud network LDP function.
Using a quadratic vs. a linear metric
is shown to improve the cost-delay tradeoff at the expense of increased computational complexity.
In addition, the use of a STPD bias yields enhanced algorithms that 
can further reduce average delay without compromising throughput
or cost performance.
\end{itemize}

The rest of the paper is organized as follows. We review related work in Section \ref{sec: related_work}. Section \ref{model} describes the system model and problem formulation. Section \ref{sec: network_capacity_region} is devoted to the characterization of the cloud network capacity region. We present the proposed DCNC algorithms in Section \ref{sec: algorithm description}, and analyze their performance in Section \ref{sec: performance_analysis}. Numerical experiments are presented in Section \ref{sec: simulation}, and possible extensions are discussed in Section \ref{sec: extensions}. Finally, we summarize the main conclusions in Section \ref{sec: conclusions}.

\section{Related Work}
\label{sec: related_work}

The problem of dynamically adjusting network resources in response to unknown changes in traffic demands has been extensively studied in previous literature in the context of stochastic network optimization.
In particular, Lyapunov drift control theory is particularly suitable for studying the stability properties of queuing networks and similar stochastic systems.
The first celebrated application of Lyapunov drift control in multi-hop networks is the backpressure (BP) routing algorithm \cite{backpressure}. The BP algorithm achieves throughput-optimality without ever designing an explicit route or having knowledge of traffic arrival rates, hence being able to adapt time-varying network conditions.
By further adding a penalty term (\eg related to network resource allocation cost) to the Lyapunov drift expression, \cite{Neely_book}-\cite{Neely_book2} developed the Lyapunov drift-plus-penalty control methodology. LDP control preserves the throughput-optimality of the BP algorithm while also minimizing average network cost.

LDP control strategies have shown effective in optimizing traditional multi-hop communication networks (as opposed to computation networks). 
Different versions of LDP-based algorithms have been developed. Most of them are based on the minimization of a linear metric obtained from an upper bound expression of the queueing system LDP function \cite{Neely_book}-\cite{Neely_book2}.
Subsequently, the inclusion of a bias term, indicative of network distance, into this linear metric was shown to reduce network delay (especially in low congested scenarios) \cite{Neely_DIVBAR}, \cite{Neely_2005}. 
Furthermore, \cite{Sucha_second_moment} proposed a control algorithm for single-commodity multi-hop networks based on the minimization of a quadratic metric from the LDP upper bound, shown to improve delay performance in the scenarios explored in \cite{Sucha_second_moment}.

In contrast to these prior works, this paper extends the LDP methodology for the dynamic control of network service chains over distributed cloud networks.
The proposed family of LDP-based algorithms 
are suitable for a general class of MCC flow problems that 
exhibit the following key features: (i) {\em Flow chaining:} a commodity, representing the flow of packets at a given stage of a service chain, can be processed into the next commodity in the service chain via the corresponding service function; (ii) {\em Flow scaling:} the flow size of a commodity can differ from the flow size of the next commodity in the service chain after service function processing; (iii) {\em Joint computation/communication scheduling:} different commodities share and compete for both processing and transmission resources, which need to be jointly scheduled.
To the best of our knowledge, this is the first attempt to address the service chain control problem in a dynamic cloud network setting.

%
%

\section{Model and Problem Formulation}
\label{model}


\subsection{Cloud Network Model}
We consider a cloud network modeled as a directed graph $\mathcal G=(\mathcal V,\mathcal E)$ with $|\mathcal V|=N$ vertices and $|\mathcal E|=E$ edges representing the set of network nodes and links, respectively. 
In the context of a cloud network, a node represents a distributed cloud location, in which virtual network functions {(VNFs)} can be instantiated in the the form of, \eg virtual machines (VMs) over general purpose servers, 
while an edge represents a logical link (\eg IP link) between two cloud locations.
We denote by $\delta^{+\!}(i)\in\mathcal V$ and $\delta^{-\!}(i)\in\mathcal V$ the set of outgoing and incoming neighbors of $i\in\mathcal V$ in $\cal G$, respectively.
We remark that in our model, cloud network nodes may represent large datacenters at the core network level, smaller edge datacenters at the metro and/or aggregation networks, or even fog \cite{fog} or cloudlet \cite{cloudlet} nodes at the access network.

We consider a time slotted system with slots normalized to integer units $t\in\{0,1,2,\cdots\}$, 
and characterize the cloud network resource capacities and costs as follows:





\begin{itemize}
\item $\mathcal K_{i}=\{0,1,\cdots,K_{i}\}$: the set of processing resource allocation choices at node $i$
\item $\mathcal K_{ij}=\{0,1,\cdots,K_{ij}\}$: the set of transmission resource allocation choices at link $(i,j)$
\item $C_{i,k}$: the capacity, in \emph{processing flow units} (\eg operations per timeslot), resulting from the allocation of $k$ processing resource units (\eg CPUs) at node $i$
\item $C_{ij,k}$: the capacity, in \emph{transmission flow units} (\eg packets per timeslot), resulting from the allocation of $k$ transmission resource units (\eg bandwidth blocks) at link $(i,j)$
\item $w_{i,k}$: the cost of allocating $k$ processing resource units at node $i$
\item $w_{ij,k}$: the cost of allocating $k$ transmission resource units at link $(i,j)$
\item $e_{i}$: the cost per processing flow unit at node $i$
\item $e_{ij}$: the cost per transmission flow unit at link $(i,j)$
\end{itemize}

\vspace{-0.1cm}

\subsection{Service Model}

A network service $\phi\in\Phi$ is described by a chain of VNFs.
We denote by $\mathcal M_\phi=\{1,2,\cdots,M_\phi\}$ the ordered set of VNFs of service $\phi$.
Hence, the pair $(\phi,m)$, with $\phi\in\Phi$ and $m\in\mathcal M_{\phi}$, identifies the $m$-th function of service $\phi$.
We refer to a client as a source-destination pair $(s,d)$, with $s,d\in\mathcal V$.
A client requesting service $\phi\in\Phi$ implies the request for the packets originating at the source node $s$ to go through the sequence of VNFs specified by $\mathcal M_\phi$ before exiting the network at the destination node $d$.


We adopt a \emph{multi-commodity-chain} (MCC) flow model, in which a commodity identifies the packets 
at a given stage of a service chain for a particular destination. Specifically, we use the triplet $(d, \phi, m)$ to identify the packets that are output of the $m$-th function of service $\phi$ for destination $d$. The source commodity of service $\phi$ for destination $d$ is denoted by $(d, \phi, 0)$ and the final commodity that are delivered to destination $d$ by $(d, \phi, M_{\phi})$, as illustrated in Fig. \ref{service_chain}.

Each VNF has (possibly) different processing requirements. 
We denote by $r^{(\phi,m)}$ the processing-transmission flow ratio of VNF $(\phi,m)$ in processing flow units per transmission flow unit (\eg operations per packet).
We assume that VNFs are fully {\em parallelizable}, in the sense that if the total processing capacity allocated at node $i$, $C_{i,k}$, 
is used for VNF $(\phi,m)$, 
then $C_{i,k}/r^{(\phi,m)}$ packets can be processed in one timeslot. 

In addition, our service model also captures the possibility of flow scaling. We denote by $\xi^{(\phi,m)}>0$ the scaling factor of VNF $(\phi,m)$, in output flow units per input flow unit. That is, the size of the output flow of VNF $(\phi,m)$ is
$\xi^{(\phi,m)}$ times as large as its input flow. We refer to a VNF with $\xi^{(\phi,m)}>1$ as an expansion function, and to a VNF with $\xi^{(\phi,m)}<1$ as a compression function.\footnote{We assume arbitrarily small
packet granularity such that arbitrary positive scaling factors can be defined.}


We remark that our service model applies to a wide range of services that go beyond NFV services, and that includes, for example, Internet of Things (IoT) services, expected to largely benefit from the proximity and elasticity of distributed cloud networks \cite{iot, iotcloud_jsac16}.


\begin{figure}
\centering
\includegraphics[width=8.8cm]{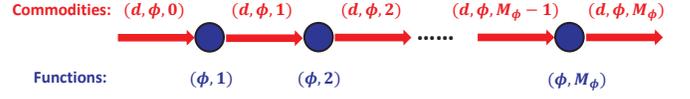}
\caption{A network service chain $\phi\in\Phi$. Service $\phi$ takes the source commodity $(d,\phi,0)$ and delivers the final commodity $(d,\phi,M_\phi)$ after going through the sequence of functions $\{(\phi,1),\cdots,(\phi,M_\phi)\}$.  VNF $(\phi,m)$ takes commodity $(d,\phi,m-1)$ and generates commodity $(d,\phi,m)$. }
\label{service_chain}
\vspace{-0.5cm}
\end{figure}

\subsection{Queuing Model}
\label{dynamic}


We denote by $a_i^{(d,\phi,m)}(t)$ the exogenous arrival rate of commodity $(d,\phi,m)$
at node $i$ during timeslot $t$, 
and by $\lambda_i^{(d,\phi,m)}$ its expected value. 
We assume that $a_i^{(d,\phi,m)}(t)$ is independently and identically distributed (i.i.d.) across timeslots, and that
$a_i^{(d,\phi,m)}(t)=0$ for $0<m\le M_\phi$, \ie there are no exogenous arrivals of intermediate commodities in a service chain.\footnote{The setting in which $a_i^{(d,\phi,m)}(t)\neq 0$ for $0<m\le M_\phi$, while of little practical relevance, does not affect the mathematical analysis in this paper.}

At each timeslot $t$, every node buffers packets according to their commodities and makes transmission and processing
{\bf flow scheduling}
decisions on its output interfaces.
Cloud network queues build up from the transmission of packets from incoming neighbors and from the local processing of packets via network service functions.
We define:
\begin{itemize}
\item $Q_i^{(d,\phi,m)}(t)$: the number of commodity $(d,\phi,m)$ packets in the queue of node $i$ at the beginning of timeslot $t$
\item $\mu_{ij}^{(d,\phi,m)}(t)$: the assigned flow rate at link $(i,j)$ for commodity $(d,\phi,m)$ at time $t$
\item $\mu_{i,\text{pr}}^{(d,\phi,m)}(t)$: the assigned flow rate from node $i$ to its processing unit for commodity $(d,\phi,m)$ at time $t$
\item $\mu_{\text{pr},i}^{(d,\phi,m)}(t)$: the assigned flow rate from node $i$'s processing unit to node $i$ for commodity $(d,\phi,m)$ at time $t$
\end{itemize}

The resulting {\em queuing dynamics} satisfies: 
\begin{align}
& \Scale[0.97]{ Q_i^{\!(d,\phi,m)}\!(t\!+\!1) \leq\! \left[ Q_i^{\!(d,\phi,m)}\!(t) \!- \!\! \displaystyle{\sum_{j\in\delta^{\!+\!}(i)}} \mu_{ij}^{\!(d,\phi,m)}\!(t) - \mu_{i,\text{pr}}^{\!(d,\phi,m)}\!(t) \! \right]^{\!+}} \notag \\
&  \Scale[0.97]{\qquad\qquad + \displaystyle{\sum_{j\in\delta^{\!-\!}(i)}} \mu_{ji}^{\!(d,\phi,m)}\!(t) + \mu_{\text{pr},i}^{\!(d,\phi,m)}\!(t) } + a_i^{\!(d,\phi,m)}\!(t),
\label{eq_queueing_dynamic}
\end{align}
where $[x]^+$ denotes $\max\{x,0\}$, and $Q_d^{(d,\phi,M_{\phi})}(t)=0, \forall d,\phi,t$.
The inequality in \eqref{eq_queueing_dynamic} is due to the fact that the actual number of packets transmitted/processed is  the minimum between the locally available packets and the  assigned flow rate.



\begin{figure}
        \includegraphics[height=3.9cm]{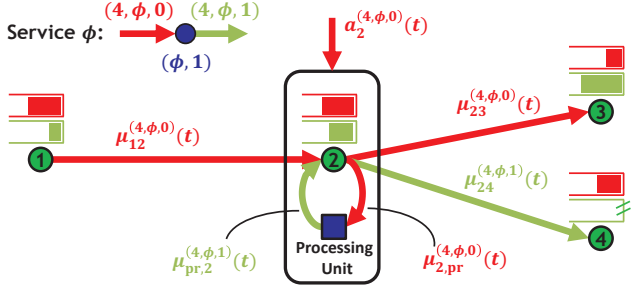}
        \centering{
        \caption{Cloud network queuing model for the delivery of a single-function service for a client with source node $2$ and destination node $4$.
        Packets of both commodity $(4,\phi,0)$ and commodity $(4,\phi,1)$ can be forwarded across the network, where they are buffered in separate commodity queues. In addition, cloud network nodes can process packets of commodity $(4,\phi,0)$ into  packets of commodity $(4,\phi,1)$, which can exit the network at node $4$.}
        \label{queuing_model}
        }
        \vspace{-0.5cm}
\end{figure}

We assume that the processing resources of node $i$ 
are co-located with node $i$ and hence the packets of commodity $(d,\phi,m-1)$ processed during timeslot $t$ are available at the queue of commodity $(d,\phi,m)$ at the beginning of timeslot $t+1$.
We can then describe the {\em service chaining dynamics} at node $i$ as follows:
\begin{align}\label{chaining}
& \mu_{\text{pr},i}^{(d,\phi,m)}(t) = \xi^{(\phi,m)} \mu_{i,\text{pr}}^{(d,\phi,m\Scale[0.7]{-1})}(t),\quad\,\,\,\, \forall d,\phi, m\!>\!0.
\end{align}

The service chaining constraints in \eqref{chaining} state that, at time $t$, the rate of commodity $(d,\phi,m)$ arriving at node $i$ from its processing unit is equal to the rate of commodity $(d,\phi,m-1)$ leaving node $i$ to its processing unit, scaled by the scaling factor $\xi^{(\phi,m)}$. Thus, Eqs. \eqref{eq_queueing_dynamic} and \eqref{chaining} imply that the packets a commodity $(d,\phi,m-1)$ processed during timeslot $t$ are available at the queue of commodity $(d,\phi,m)$ at the beginning of timeslot $t+1$.

As an example, the cloud network queuing system of an illustrative $4$-node cloud network is shown in Fig. \ref{queuing_model}.


%
%

In addition to processing/transmission flow scheduling decisions, at each timeslot $t$, cloud network nodes can also make transmission and processing {\bf resource allocation} decisions. We use the following binary variables to denote the resource allocation decisions at time $t$:
\begin{itemize}
\item $y_{i,k}(t)=1$ if $k$ processing resource units are allocated at node $i$ at time $t$; $y_{i,k}(t)=0$, otherwise
\item $y_{ij,k}(t)=1$ if $k$ transmission resource units are allocated at link $(i,j)$ at time $t$; $y_{ij,k}(t)=0$, otherwise
\end{itemize}



\subsection{Problem Formulation}

The goal is to design a dynamic control policy, defined by a flow scheduling and resource allocation action vector $\{\bm \mu(t), \yv(t)\}$, that supports all average input rate matrices
$\bm \lambda \triangleq \{\lambda_i^{(d,\phi,m)}\}$ that are interior to the \emph{cloud network capacity region} (as defined in Section \ref{sec: network_capacity_region}),
while minimizing the total average cloud network cost.
Specifically, we require the cloud network to be \emph{rate stable} (see Ref. \cite{Neely_book}), \ie
\begin{align}
&  \lim_{t\rightarrow\infty} \frac{Q_i^{(d,\phi,m)}(t)}{t} = 0 \quad \text{with prob. 1}, \qquad \forall i, d,\phi,m. \label{stab}
\end{align}

The dynamic cloud network control problem 
can then be formulated as follows:
\begin{subequations}\label{opt2}
\begin{align}
&\text{min} && \!\! \underset{t\rightarrow\infty}{\limsup} \quad  \frac{1}{t} \, \sum\nolimits_{\tau=0}^{t-1} \, \mathbb{E}\left\{h(\tau)\right\} \label{obj2}\\
&\text{s.t.} && \!\! \text{The cloud network is rate stable}, \label{stab2}\\
&&& \!\! \mu_{\text{pr},i}^{(d,\phi,m)}\!(\tau)\! =\! \xi^{(\phi,m)}\! \mu_{i,\text{pr}}^{(d,\phi,m\Scale[0.6]{-1})}\!(\tau),\ \ \ \forall i, d, \phi, m, \tau, \! \label{chain2}\\
&&& \!\!\!\!\! \sum_{(d,\phi,m)} \!\! \mu_{i,\text{pr}}^{(d,\phi,m)}\!(\tau) r^{(\phi,m\Scale[0.6]{+1})} \! \leq \! \sum_{k\in\mathcal K_i}\!\! C_{i,k} \, y_{i,k}(\tau), \, \forall i,\tau, \! \label{cappr2} \\
&&& \!\!\!\!\! \sum_{(d,\phi,m)} \!\! \mu_{ij}^{(d,\phi,m)}(\tau) \! \leq \! \sum_{k\in\mathcal K_{ij}} \!\! C_{ij,k} \, y_{ij,k}(\tau), \,\,\,\, \forall (i,j), \tau, \! \label{captr2}\\
&&& \!\! \mu_{i,\text{pr}}^{(d,\phi,m)}\!(\tau),\ \mu_{\text{pr},i}^{(d,\phi,m)}\!(\tau),\ \mu_{ij}^{(d,\phi,m)}\!(\tau)\! \in {\mathbb R}^+,  \notag\\
&&& \!\! \qquad\qquad\qquad\qquad\qquad\qquad\,\,\, \forall i,(i,j), d, \phi, m, \tau, \label{fs}\\
&&& \!\! y_{i,k}(\tau),\ y_{ij,k}(\tau) \in\{0,1\},   \quad\,\,\,  \forall i,(i,j), d,\phi, m, \tau, \! \label{ys}
\end{align}
\end{subequations}
where $h(\tau)\triangleq\sum_{i\in \mathcal V}{h_i(\tau)}$, with 
\begin{align}
\label{obj}
& h_i(\tau) = \sum\limits_{k \in {{\cal K}_{i}}}\! {{w_{i,k}}} {y_{i,k}}(\tau) + {e_i} \!\! \sum_{(d,\phi,m)} \!\!\! \mu_{i,\text{pr}}^{(d,\phi,m)}\!(\tau) r^{(\phi,m\Scale[0.65]{+1})} \notag \\
& \quad +\!\!  \sum\limits_{j\in \delta^{\!+\!}(i)} \!\left[ \sum\limits_{k \in \mathcal K_{ij}} w_{ij,k} y_{ij,k}(\tau) + e_{ij} \!\! \sum_{(d,\phi,m)} \!\! \mu_{ij}^{(d,\phi,m)}\!(\tau) \right] \!,\!
\end{align}
denotes the cloud network operational cost at time $\tau$.



In \eqref{opt2},
Eqs. \eqref{chain2},  \eqref{cappr2}, and \eqref{captr2} describe instantaneous service chaining, processing capacity, and transmission capacity constraints, respectively. 

\begin{rem}
As in Eqs. \eqref{chain2},  \eqref{cappr2}, \eqref{obj}, throughout this paper, it shall be useful to establish relationships between consecutive commodities and/or functions in a service chain. For ease of notation, unless specified, we shall assume that any expression containing a reference to $m-1$ will only be applicable for $m>0$ and any expression with a reference to $m+1$ will only be applicable for $m<M_{\phi}$.
\end{rem}

In the following section, we characterize the cloud network capacity region in terms of the average input rates that can be stabilized by any control algorithm that satisfies constraints \eqref{stab2}-\eqref{ys}, as well as the minimum average cost required for cloud network stability.


%
%



\section{Cloud Network Capacity Region}
\label{sec: network_capacity_region}

The cloud network capacity region $\Lambda(\mathcal G, \Phi)$ is defined as the closure of all average input rates $\bm \lambda$  that can be stabilized by a cloud network control algorithm, whose decisions conform to the cloud network and service structure $\{\mathcal G, \Phi\}$. 

\begin{thm}\label{thm: network_capacity_region}
The cloud network capacity region $\Lambda(\mathcal G, \Phi)$ consists of all average input rates $\bm \lambda$ for which, for all $i,j,k,d,\phi,m$, there exist MCC flow variables $f_{ij}^{(d,\phi,m)}$, $f_{\text{pr},i}^{(d,\phi,m)}$, $f_{i,\text{pr}}^{(d,\phi,m)}$, together with probability values $\alpha_{ij,k}$, $\alpha_{i,k}$, $\beta_{ij,k}^{(d,\phi,m)}$, $\beta_{i,k}^{(d,\phi,m)}$ 
such that

\begin{subequations}
\begin{align}
& \sum\limits_{j \in \delta^{\!-\!} \left( i \right)} \!\!\! {f_{ji}^{\left( {d,\phi ,m} \right)}} \!\! + f_{\emph{pr},i}^{\left( {d,\phi ,m} \right)} \!\! +
\lambda_i^{(d,\phi,m)} \! \le \!\!\sum\limits_{j \in \delta^{\!+\!} \left( i \right)} \!\!\! {f_{ij}^{\left( {d,\phi ,m} \right)}}  \!\! + f_{i,\emph{pr}}^{\left( {d,\phi ,m} \right)},  \label{eq_thm1_stability} \\
& f_{\emph{pr},i}^{\left( {d,\phi ,m} \right)}=\xi^{(\phi,m)}f_{i,\emph{pr}}^{\left( {d,\phi ,m-1} \right)}, 
\label{eq_thm1_processing_conservation}\\
& f_{i,\emph{pr}}^{\left( {d,\phi ,m} \right)} \le \frac{1}{r^{(\phi ,m\Scale[0.7]{+1})}} \sum\nolimits_{k \in {{\cal K}_{i}}} {{\alpha _{i,k}} \beta _{i,k}^{\left( {d,\phi ,m} \right)} {C_{i,k}}},
 \label{eq_thm1_processing_capacity_constraint} \\
& f_{ij}^{\left( {d,\phi ,m} \right)} \le \sum\nolimits_{k \in {{\cal K}_{ij}}} {{\alpha _{ij,k}}\beta _{ij,k}^{\left( {d,\phi ,m} \right)}{C_{ij,k}}}, \label{eq_thm1_flow_capacity_constraint}\\
& f_{i,\emph{pr}}^{\left( {d,\phi ,{M_\phi }} \right)} = 0,\ f_{\emph{pr},i}^{\left( {d,\phi ,0} \right)} = 0,\ f_{dj}^{\left( {d,\phi ,{M_\phi }} \right)} = 0, \label{eq_thm1_flow_boundary_conditions} \\
& f_{i,\emph{pr}}^{\left( {d,\phi,m} \right)} \geq 0,\ f_{ij}^{\left( {d,\phi ,m} \right)} \geq 0,  \label{eq_thm1_flow_positve_conditions} \\
& \sum\nolimits_{k \in {{\cal K}_{ij}}} {{\alpha _{ij,k}}} \le 1, \sum\nolimits_{k \in {{\cal K}_{i}}} {{\alpha _{i,k}}} \le 1,  \label{alphas}  \\
& \sum\nolimits_{\left( {d,\phi ,m} \right)} {\beta _{ij,k}^{\left( {d,\phi ,m} \right)}}  \le 1, \sum\nolimits_{\left( {d,\phi ,m} \right)} {\beta _{i,k}^{\left( {d,\phi ,m} \right)}}  \le 1.  \label{betas}
\end{align}
\end{subequations}

Furthermore, the minimum average cloud network cost required for network stability is given by
\begin{align}
\label{eq_thm1_minimum_cost}
\overline h^* = \min_{\{\alpha_{ij,k}, \alpha_{i,k}, \beta_{ij,k}^{(d,\phi,m)}, \beta_{i,k}^{(d,\phi,m)}\}} \overline h,
\end{align}
 where
\begin{align}
\label{eq:underline}
\overline h &= \sum\limits_{i} \sum\limits_{k \in {\cal K}_{i}} {\alpha _{i,k}} \left( w_{i,k} + e_i C_{i,k} \!\! \sum\limits_{\left( {d,\phi ,m} \right)} \!\! \beta _{i,k}^{\left( {d,\phi ,m} \right)} \right) \nonumber\\
&+\sum\limits_{(i,j)} {\sum\limits_{k \in {{\cal K}_{ij}}} {{\alpha _{ij,k}}\left( {{w_{ij,k}} + {e_{ij}}{C_{ij,k}} \! \sum\limits_{\left( {d,\phi ,m} \right)} \! {\beta _{ij,k}^{\left( {d,\phi ,m} \right)}} } \right)} }.
\end{align}
\end{thm}

\begin{proof}
The proof of Theorem \ref{thm: network_capacity_region} is given by Appendix \ref{Proof of capacity_region_nessesary} .

\end{proof}


In Theorem \ref{thm: network_capacity_region},
\eqref{eq_thm1_stability} and \eqref{eq_thm1_processing_conservation} describe generalized computation/communication flow conservation constraints and service chaining constraints, essential for cloud network stability, while
\eqref{eq_thm1_processing_capacity_constraint} and \eqref{eq_thm1_flow_capacity_constraint} describe processing and transmission capacity constraints.
The probability values $\alpha_{i,k}$, $\alpha_{ij,k}$, $\beta_{i,k}^{(d,\phi,m)}$,  $\beta_{ij,k}^{(d,\phi,m)}$
define a \emph{stationary randomized policy} as follows: 
\begin{itemize}
\item  $\alpha_{i,k}$: the probability that $k$ processing resource units  are allocated at node $i$;
\item  $\alpha_{ij,k}$: the probability that $k$ transmission resource units are allocated at link $(i,j)$;
\item  $\beta_{i,k}^{(d,\phi,m)}$: the probability that node $i$ processes commodity $(d,\phi,m)$, conditioned on the allocation of $k$ processing resource units at node  $i$;
\item  $\beta_{ij,k}^{(d,\phi,m)}$: the probability that link $(i,j)$  transmits commodity $(d,\phi,m)$, conditioned on the allocation of $k$ transmission resource units at link $(i,j)$.
\end{itemize}


Hence, Theorem \ref{thm: network_capacity_region} demonstrates that, for any input rate $\bm \lambda \in \Lambda(\mathcal G, \Phi)$, there exists a stationary randomized policy 
that uses fixed probabilities to make transmission and processing decisions at each timeslot, which can support the given $\bm \lambda$, 
while minimizing overall average cloud network cost.
However, the difficulty in directly solving for the parameters that characterize such a stationary randomized policy and the requirement on the knowledge of
$\bm \lambda$, motivates the design of online dynamic cloud network control solutions with matching performance guarantees. 

\section{Dynamic Cloud Network Control Algorithms}
\label{sec: algorithm description}



In this section, we describe distributed DCNC strategies that account for both processing and transmission flow scheduling and resource allocation decisions. We first propose  DCNC-L, an algorithm based on minimizing a \emph{linear metric} obtained from an upper bound of the quadratic LDP function, where only linear complexity is required for making local decisions at each timeslot. We then propose DCNC-Q, derived from the minimization of a \emph{quadratic metric} obtained from the LDP bound. DCNC-Q allows simultaneously scheduling multiple commodities on a given transmission or processing interface at each timeslot, leading to a more balanced system evolution that can improve the cost-delay tradeoff at the expense of quadratic computational complexity.
Finally, enhanced versions of the aforementioned algorithms, referred to as EDCNC-L and EDCNC-Q, are constructed by adding a shortest transmission-plus-processing distance (STPD) bias extension that is shown to further reduce network delay in low congested scenarios.

\subsection{Cloud Network Lyapunov drift-plus-penalty}
\label{sec: LDP}

Let ${\bf Q}(t)$ represent the vector of queue backlog values of all the commodities at all the cloud network nodes. 
The cloud network \emph{Lyapunov drift} is defined as
\begin{equation}
\Delta \left( {{\bf{Q}}\left( t \right)} \right) \triangleq \frac{1}{2}\mathbb{E}\left\{ {\left. {{{\left\| {{\bf{Q}}\left( {t + 1} \right)} \right\|}^2} - {{\left\| {{\bf{Q}}\left( t \right)} \right\|}^2}} \right|{\bf{Q}}\left( t \right)} \right\},
\end{equation}
where $\|\cdot\|$ indicates Euclidean norm, and the expectation is taken over the ensemble of all the exogenous source commodity arrival realizations.

The one-step Lyapunov drift-plus-penalty (LPD) is then defined as 
\begin{align}
\label{eq_ldp}
\Delta &\left( {{\bf{Q}}\left( t \right)} \right)+ V\mathbb{E}\left\{ {\left. {{h}(t)} \right|{\bf{Q}}\left( t \right)} \right\} ,
\end{align}
where $V$ is a non-negative control parameter that determines the degree to which resource cost minimization is
emphasized.

After squaring both sides of \eqref{eq_queueing_dynamic} and following standard LDP manipulations (see Ref. \cite{Neely_book2}), the LDP can upper bound  as
\begin{align}
\label{eq_lypunov_bound1}
\Delta &\left( {{\bf{Q}}\left( t \right)} \right)+ V\mathbb{E}\left\{ {\left. {{h}(t)} \right|{\bf{Q}}\left( t \right)} \right\} \le V\mathbb{E}\left\{ {\left. {{h}(t)} \right|{\bf{Q}}\left( t \right)} \right\}+\nonumber\\
& \mathbb{E}\left\{\left. \Gamma(t) + Z(t) \right|{\bf Q}(t)\right\} + \sum_i\!\!\!\sum\limits_{\left( {d,\phi ,m} \right)}\!\!\!{\lambda_i^{(d,\phi,m)}Q_i^{d,\phi,m}(t)},
\end{align}
where
\begin{align}
&\Gamma(t) \triangleq\frac{1}{2}  \sum_i \! \sum\limits_{\left( {d,\phi ,m} \right)} \!{\left\{ {{{\left[ {\sum\limits_{j\in \delta^{\!+\!}(i)} {\mu _{ij}^{(d,\phi ,m)}(t)}  + \mu _{i,\text{pr}}^{(d,\phi ,m)}(t)} \right]}^2}} \right.} \nonumber\\
&\qquad\ \left. { + {{\left[ {\sum_{j\in \delta^{\!-\!}(i)}  {\mu _{ji}^{(d,\phi ,m)}\!(t)}  + \mu _{\text{pr},i}^{(d,\phi ,m)}\!(t) + a_i^{\left( {d,\phi ,m} \right)}\!\left( t \right)} \right]}^2}}\! \right\}\!,\nonumber\\
&Z(t)\triangleq \sum_i \! \sum\limits_{\left( {d,\phi ,m} \right)} {Q_i^{\left( {d,\phi ,m} \right)}\!\left( t \right)\left[ {\sum_{j\in \delta^{\!-\!}(i)}  {\mu _{ji}^{(d,\phi ,m)}(t)} } \right.} \nonumber\\
&\qquad\quad\,\,\,\, \left. { + \, \mu _{\text{pr},i}^{(d,\phi ,m)}\!(t) - \! \sum_{{j\in \delta^{\!+\!}(i)}  } \! {\mu _{ij}^{(d,\phi ,m)}(t)}  - \mu _{i,\text{pr}}^{(d,\phi ,m)}(t)} \right]\!\!.\nonumber
\end{align}

Our DCNC algorithms extract different metrics from the right hand side of \eqref{eq_lypunov_bound1}, 
whose minimization leads to a family of throughput-optimal flow scheduling and resource allocation policies with different cost-delay tradeoff performance.

\subsection{Linear Dynamic Cloud Network Control (DCNC-L)}
\label{subsec: alg_dcnc_l}

DCNC-L is designed to minimize, at each timeslot, the linear metric $Z(t)+Vh(t)$ obtained from the right hand side of \eqref{eq_lypunov_bound1}, equivalently expressed as 
\begin{subequations}\label{algor}
\begin{align}
&\text{min} \!\!   &&\!\!\sum\limits_{i \in {\cal V}}\! {\left[\! {V{h_i}(t)\!- \!\!\! \sum\limits_{\left( {d,\phi ,m} \right)} {\!\!\! \left(\sum\limits_{j \in \delta^{\!+\!}(i)} {Z_{ij,\text{tr}}^{\left( {d,\phi ,m} \right)}(t) + Z_{i,\text{pr}}^{\left( {d,\phi ,m} \right)}(t)}\!\! \right)} } \!\right]} \label{obj3}\\
&\text{s.t.} && \eqref{cappr2}-\eqref{ys},
\end{align}
\end{subequations}
where,
\begin{align}
& \Scale[1.01]{ Z_{ij,\text{tr}}^{\left( {d,\phi ,m} \right)}\!(t) \!\triangleq\! \mu _{ij}^{(d,\phi ,m)}\!(t)\!\left[ {Q_i^{(d,\phi ,m)}\!(t) - Q_j^{(d,\phi ,m)}\!(t)} \right]  },  \notag  \\
& \Scale[1.01] {Z_{i,\text{pr}}^{\left( {d,\phi ,m} \right)}\!(t) \!\triangleq\! \mu _{i,\text{pr}}^{(d,\phi ,m)}\!(t)\!\!\left[ {Q_i^{(d,\phi ,m)}\!(t) - {\xi ^{(\phi,m\Scale[0.6]{+1} )}}Q_i^{(d,\phi ,m\Scale[0.6]{+1})}\!(t)}\! \right]}\!. \notag
\end{align}

The goal of minimizing \eqref{obj3} at each timeslot is to greedily push the cloud network queues 
towards a lightly congested state, while minimizing cloud network resource usage regulated by the control parameter $V$. Observe that \eqref{obj3} is a linear metric with respect to $\mu_{i,\text{pr}}^{(d,\phi,m)}(t)$ and $\mu_{ij}^{(d,\phi,m)}(t)$, and hence \eqref{algor} can be decomposed into the implementation of \emph{Max-Weight-Matching} \cite{Algorithm_Kleignberg} at each node, leading to the following distributed flow scheduling and resource allocation policy:

\underline{Local processing decisions:}
At the beginning of each timeslot $t$, each node $i$ observes its local queue backlogs and performs the following operations:

\begin{enumerate}

\item Compute the \emph{processing utility weight} of each processable commodity, $(d,\phi,m), m<M_{\phi}$:
\begin{align}
\Scale[1]{\!\! W_i^{(d,\phi,m)}\! (t) \!=\!\! \left[\!\frac{ Q_i^{(d,\phi,m)}\!(t) -  \xi^{(\phi,m\Scale[0.6]{+1})}Q_i^{(d,\phi,m+1)}\!(t) }{r^{(\phi,m+1)}}\!-\! V e_{i}  \!\right]^{\!+} }\!, \notag
\end{align}
and set $W_i^{(d,\phi,M_{\phi})}\!(t)=0, \forall d,\phi$.
$W_i^{(d,\phi,m)}(t)$ is indicative of the potential benefit of processing commodity $(d,\phi,m)$ into commodity $(d,\phi,m\!+\!1)$ at time $t$, in terms of the difference between local congestion reduction 
and processing cost per unit flow.


\item Compute the max-weight commodity: 
\begin{align}
\Scale[1]{ (d,\phi,m)^* =   \underset{(d,\phi,m)}{\arg\max} \left\{ W_i^{(d,\phi,m)}{(t)} \right\} }. \notag
\end{align}

\item
If $W_i^{(d,\phi,m)^*}\!\!{(t)}=0$, set, $k^*=0$.
Otherwise,
\begin{align}
\Scale[1]{ k^*  =  \underset{k}{\arg\max} \left\{ C_{i,k} W_i^{(d,\phi,m)^*}\!\!{(t)} - V w_{i,k}  \! \right\} }. \notag
\end{align}


\item Make the following resource allocation and flow assignment decisions:
\begin{align}
& y_{i,k^*}(t)=1,\notag\\
&y_{i,k}(t)=0,\quad \forall k\neq k^*, \notag\\
& \mu_{i,\text{pr}}^{(d,\phi,m)^*}(t) =
C_{i,{k^*}}\left/r^{(\phi ,m+1)^*}\right.,\notag\\
& \mu_{i,\text{pr}}^{(d,\phi,m)}(t) = 0,\quad\forall (d,\phi,m)\neq (d,\phi,m)^*. \notag
\end{align}

\end{enumerate}

\underline{Local transmission decisions:}
At the beginning of each timeslot $t$, each node $i$ observes its local queue backlogs and those of its neighbors, and performs the following operations for each of its outgoing links $(i,j)$, $j\in\delta^{\!+\!}(i)$:

\begin{enumerate}

\item Compute the \emph{transmission utility weight} of each commodity $(d,\phi,m)$:
\begin{align}
{W_{ij}^{(d,\phi,m)}\left(t\right) }  = \left[ Q_i^{(d,\phi,m)}(t) - Q_j^{(d,\phi,m)}(t) - V e_{ij}  \right]^+. \notag
\end{align}
\item Compute the max-weight commodity: 
\begin{align}
(d,\phi,m)^* =  \underset{(d,\phi,m)}{\arg\max} {\left\{ W_{ij}^{(d,\phi,m)}(t) \right\} }. \notag
\end{align}

\item If ${W_{ij}^{(d,\phi,m)^*}\left(t\right)}=0$, set, $k^*=0$.
Otherwise,
\begin{align}
k^* =  \underset{k}{\arg\max} \left\{ C_{ij,k} W_{ij}^{(d,\phi,m)^*} \!\! {(t)}  - V w_{ij,k}  \! \right\}. \notag
\end{align}

\item Make the following resource allocation and flow assignment decisions:
\begin{align}
& y_{ij,k^*}(t)=1,\notag\\
& y_{ij,k}(t)=0,\quad \forall k\neq k^*, \notag\\
& \mu_{ij}^{(d,\phi,m)^*}(t) = C_{ij,k^*}, \notag\\
& \mu_{ij}^{(d,\phi,m)}(t) = 0 \quad \forall (d,\phi,m)\neq (d,\phi,m)^*. \notag
\end{align}

\end{enumerate}

Implementing the above algorithm imposes low complexity on each node. Let $J$ denote the total number of commodities. We have $J\le N\sum\nolimits_\phi{(M_\phi+1)}$. Then, the total complexity associated with the processing and transmission decisions of node $i$ at each timeslot is $O(J+K_i + \sum\nolimits_{j\in \delta^+(i)}K_{ij})$, which is linear with respect to the number of commodities and the number of resource allocation choices.

\begin{rem}
Recall that, while assigned flow values can be larger than the corresponding queue lengths, a practical algorithm will only send those packets available for transmission/processing. However, as in \cite{Neely_book}-\cite{Sucha_second_moment}, in our analysis, we assume a policy that meets assigned flow values with null packets (\eg filled with idle bits) when necessary. Null packets consume resources, but do not build up in the network.
\end{rem}


\subsection{Quadratic Dynamic Cloud Network Control (DCNC-Q)}

DCNC-Q is designed to minimize, at each timeslot, the metric formed by the sum of the quadratic terms $(\mu_{ij}^{(d,\phi,m)}(t))^2$, $(\mu_{i,\text{pr}}^{(d,\phi,m)}(t))^2$, and $(\mu_{\text{pr},i}^{(d,\phi,m)}(t))^2$, extracted from
$\Gamma(t)$, and $Z(t)+Vh(t)$, on the right hand side of \eqref{eq_lypunov_bound1}, equivalently expressed as
\begin{subequations}\label{algor_Q}
\begin{align}
&\text{min}\!\!\!\!  && \displaystyle \sum\limits_{i \in {\cal V}} {\left\{ \displaystyle  \sum_{\left( {d,\phi ,m} \right)} \sum\limits_{j \in \delta^{\!+\!} (i)} {\left[ {{{\left( {\mu _{ij}^{(d,\phi ,m)}\!(t)} \right)}^2} - Z_{ij,\text{tr}}^{\left( {d,\phi ,m} \right)}\!(t)} \right]}   \right.} \nonumber\\
&&& +\!\!  \sum_{\left( {d,\phi ,m} \right)} \!{\left[ \frac{{1 \!+\! {{\left( {{\xi ^{(\phi ,m + 1)}}} \right)}^2}}}{2}{{\left( \! {\mu _{i,\text{pr}}^{(d,\phi ,m)}\!(t)} \! \right)}^{\!2}} \! - Z_{i,\text{pr}}^{\left( {d,\phi ,m} \right)}\!(t) \right]} \nonumber\\
&&& \Big.  + \, V{h_i}(t)  \Big\}
\label{obj_Q}\\
&\text{s.t.} && \eqref{cappr2}-\eqref{ys}.
\end{align}
\end{subequations}


The purpose of \eqref{algor_Q} is also to reduce the congestion level while minimizing resource cost. However, by introducing the quadratic terms $(\mu_{i,\text{pr}}^{(d,\phi,m)}(t))^2$ and $(\mu_{ij}^{(d,\phi,m)}(t))^2$, minimizing \eqref{obj_Q} results in a ``smoother'' and more ``balanced'' flow and resource allocation solution, which has the potential of improving the cost-delay tradeoff, with respect to the max-weight solution of DCNC-L that allocates either zero or full capacity to a single commodity at each timeslot.
Note that \eqref{algor_Q} can also be decomposed into subproblems at each cloud network node. Using the \emph{KKT conditions} \cite{optimization_boyd}, the solution to each subproblem admits a simple waterfilling-type interpretation. We first describe the resulting 
local flow scheduling and resource allocation policy and then provide its graphical interpretation.

\underline{Local processing decisions:} At the beginning of each timeslot $t$, each node $i$ observes its local queue backlogs and performs the following operations:
\begin{enumerate}
\item Compute the processing utility weight of each commodity. 
Sort the resulting set of weights 
in non-increasing order and form the list $\{ W_i^{(c)}(t)\}$, where $c$ identifies the $c$-th commodity in the sorted list.

\item
For each resource allocation choice $k\in \mathcal K_{i}$:

2.1) Compute the following {\em waterfilling rate threshold}:

\begin{equation}
G_{i,k}(t) \triangleq \left[ \frac{ \sum\limits_{s = 1}^{p_{k}} \frac{(r^{(s)})^2}{1+(\xi^{(s)})^2} W_{i}^{(s)}(t) - C_{i,k}} {\sum\limits_{s = 1}^{p_{k}} \frac{(r^{(s)})^2}{1+(\xi^{(s)})^2}} \right]^+, \notag
\end{equation}
where $p_{k}$ is the smallest commodity index that satisfies $H_{i}^{(p_{k})}(t) > C_{i,k}$, with $p_{k}=J$ if $C_{i,k}\geq H_i^{(J)}(t)$;
and
\begin{eqnarray}
    &H_i^{(c)} \! \left( t \right) \!\triangleq\! \sum_{s = 1}^c \!{\left[ {{W_i^{{{\left( s \right)}}}\left( t \right)}-{{W_{i}^{{{(c+1)}}}\left( t \right)}}} \right]} \frac{(r^{(s)})^2}{1+(\xi^{(s)})^2} \, , \notag 
    \end{eqnarray}
with $r^{(s)}$ and $\xi^{(s)}$ denoting the processing-transmission flow ratio and the scaling factor of the function that processes commodity $s$, respectively.

2.2) Compute the \emph{candidate} processing flow rate for each commodity, $1\leq c \leq J$: 
     \begin{align}
    &\mathord{\buildrel{\lower3pt\hbox{$\scriptscriptstyle\smile$}}
    \over \mu} _{i,\textrm{pr}}^{{{\left( c \right)}}}\!\left( k,t \right) = \frac{r^{(c)}}{1+(\xi^{(c)})^2} \left[W_{i}^{{{\left( c \right)}}} \!\left( t \right) - G_{i,k}(t)\right]^+ . \notag
    \end{align}

2.3) Compute the following optimization metric: 
    \begin{align}
    & \Psi _i(k,t) \triangleq \sum_{c=1}^{J} \left[ \frac{1+(\xi^{(c)})^2}{2} \left( \mathord{\buildrel{\lower3pt\hbox{$\scriptscriptstyle\smile$}}\over \mu}_{i,\text{pr}}^{(c)}(k,t) \right)^2 \right. \notag\\
    &\qquad\qquad\qquad \left. - \mathord{\buildrel{\lower3pt\hbox{$\scriptscriptstyle\smile$}}\over \mu}_{i,\text{pr}}^{(c)}(k,t) r^{(c)} W_{i,\text{pr}}^{(c)}(t) \right] + V w_{i,k}. \nonumber
    \end{align}
\item Compute the processing resource allocation choice: 
    \begin{equation}
    {k^*} = \mathop {\arg \min }\nolimits_{k \in {\mathcal K_{i}}} \left\{ {{\Psi _{i}}\left( k,t \right)} \right\}.\nonumber
    \end{equation}
\item Make the following resource allocation and flow assignment decisions:
    \begin{align}
    &y_{i,k^*}(t) = 1,\nonumber\\
    &y_{i,k^*}(t) = 0, {\quad \text{for}\ }k\ne k^*, \nonumber\\
    &\mu _{i,\text{pr}}^{(c)}(t) = \mathord{\buildrel{\lower3pt\hbox{$\scriptscriptstyle\smile$}}
    \over \mu} _{i,\text{pr}}^{(c)}(k^*,t). \notag 
    \end{align}

\end{enumerate}

\underline{Local transmission decisions:}
At the beginning of each timeslot $t$, each node $i$ observes its local queue backlogs and those of its neighbors, and performs the following operations for each of its outgoing links $(i,j)$, $j\in\delta^{\!+\!}(i)$:
\begin{enumerate}
\item

    Compute the transmission utility weight of each commodity. 
    Sort the resulting set of weights 
    in non-increasing order and form the list $\{ W_{ij}^{(c)}(t)\}$, where $c$ identifies the $c$-th commodity in the sorted list.


\item For each resource allocation choice $k\in \mathcal K_{i}$:

2.1) Compute the following {\em waterfilling rate threshold}:
    \begin{equation}
    G_{ij,k}(t)\triangleq \frac{1}{p_k}
    \left[ {\sum\nolimits_{s = 1}^{{p_{k}}}\! {W_{ij}^{\left( s \right)}\!\!\left( t \right)} -  2{C_{ij,k}}} \right]^+.\notag
    \end{equation}
    where $p_{k}$ is the smallest commodity index that satisfies $H_{ij}^{(p_{k})}(t) > C_{ij,k}$, with $p_{k}=J$ if $C_{ij,k}\geq H_i^{(J)}(t)$; and
    \begin{eqnarray}
    &\!\!\!\!H_{ij}^c\left( t \right) \!\triangleq\!
    \frac{1}{2}\sum\limits_{s = 1}^c {\left[ {W_{ij}^{{{\left( s \right)}}}\!\!\left( t \right)\! -\! W_{ij}^{{\left( c+1 \right)}}\!\!\left( t \right)} \right]}\!.\notag
    \end{eqnarray}

2.2) Compute the candidate transmission flow rate for each commodity, $1\leq c \leq J$: 
    \begin{align}
    &\mathord{\buildrel{\lower3pt\hbox{$\scriptscriptstyle\smile$}}
    \over \mu}_{ij}^{(c)}\!(k,t) = \frac{1}{2}\left[ W_{ij}^{{\left( c \right)}}\!\left( t \right) -  G_{ij,k}(t)\right]^{\!+}. \notag 
    \end{align}

2.3) Compute the following optimization metric: 
    \begin{align}
    &{{\Psi _{ij}}\!\left( k,t \right) \!\triangleq\!  Vw_{ij,k}\!+\!\sum\limits_{c=1}^{J}\!\! {\left[\! {{{\left(\! {\mathord{\buildrel{\lower3pt\hbox{$\scriptscriptstyle\smile$}}
    \over \mu} _{ij}^{(c)}\!(k,t)}\! \right)}^{\!2}}\!\! -\!\! \mathord{\buildrel{\lower3pt\hbox{$\scriptscriptstyle\smile$}}
    \over \mu} _{ij}^{(c)}\!(k,t)W_{ij}^{(c)}\!(t)}\! \right]}\!.}\nonumber
    \end{align}

\item  Compute the processing resource allocation choice: 
    \begin{equation}
    {k^*} = \mathop {\arg \min }\nolimits_{k \in {\mathcal K_{ij}}} \left\{ {{\Psi _{ij}}\left( k,t \right)} \right\}.\nonumber
    \end{equation}
\item Make the following resource allocation and flow assignment decisions:
    \begin{align}
    &y_{ij,k^*}(t) = 1,\nonumber\\
    &y_{ij,k}(t) = 0, \quad \forall k\ne k^*,\nonumber\\
    &\mu _{ij}^{(c)}(t) = \mathord{\buildrel{\lower3pt\hbox{$\scriptscriptstyle\smile$}}
    \over \mu} _{ij}^{(c)}(k^*,t).\nonumber
    \end{align}
\end{enumerate}

The total complexity is $O(J[\log_2 J + K_i + \sum\nolimits_{j\in \delta^+(i)}{K_{ij}}])$, which is quadratic respective to the number of commodities and the number of resource allocation choices.

As stated earlier, DCNC-Q admits a waterfilling-type interpretation, illustrated in Fig. \ref{water_filling}. We focus on the local processing decisions. Define a two-dimensional vessel for each commodity. The height of vessel $c$ is given by the processing utility weight of commodity $c$, $W_i^{(c)}(t)$, and its width by $\frac{(r^{(c)})^2}{1+(\xi^{(c)})^2}$. For each resource allocation choice $k\in\mathcal K_i$, pour mercury on each vessel up to height $G_{i,k}(t)$ given in step 2.1 (indicated with yellow in the figure). If available, fill the remaining of each vessel with water (blue in the figure). The candidate assigned flow rate of each commodity is given by the amount of water on each vessel (step 2.2), while to total amount of water is equal to the available capacity $C_{i,k}$.  Finally, step 3 is the result of choosing the resource allocation choice $k^*$ that minimizes $\eqref{obj_Q}$ with the corresponding assigned flow rate values. The local transmission decisions follow a similar interpretation that is omitted here for brevity.

\begin{figure}
        \includegraphics[height=4cm]{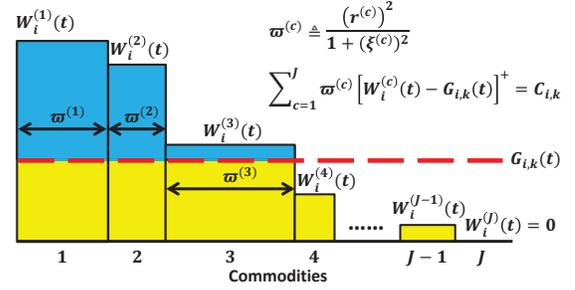}
        \centering{
        \caption{Waterfilling interpretation of the local processing decisions of DCNC-Q at time $t$.}
        \label{water_filling}}
        \vspace{-0.5cm}
\end{figure}

\vspace{-0.4cm}
\subsection{Dynamic Cloud Network Control with Shortest Transmi\-ssion-plus-Processing Distance Bias}

DCNC algorithms determine packet routes and processing locations according to the evolution of the cloud network commodity queues. However, queue backlogs have to build up before yielding efficient processing and routing configurations, which can result in degraded delay performance, especially in low congested scenarios.


In order to reduce average cloud network delay, we extend the approach used in \cite{Neely_DIVBAR}, \cite{Neely_2005} for traditional communication networks, which consists of incorporating a bias term into the metrics that drive scheduling decisions. In a cloud network setting, this bias is designed to capture the delay penalty incurred by each forwarding and processing operation. 

Let $\hat Q_i^{(d,\phi,m)}(t)$ denote the {\em biased} backlog of commodity $(d,\phi,m)$ at node $i$:
\begin{equation}
\hat Q_i^{(d,\phi,m)}(t) \triangleq Q_i^{(d,\phi,m)}(t) + \eta Y_i^{(d,\phi,m)},\label{eq_abstract_backlog}
\end{equation}
where $Y_i^{\left( {d,\phi ,m} \right)}$ denotes the \emph{shortest transmission-plus-processing distance bias} (STPD), and $\eta$ is a control parameter used to balance the effect of the bias and the queue backlog.
The bias term in \eqref{eq_abstract_backlog} is defined as
\begin{equation}
\label{eq_bias_term}
Y_i^{\left( {d,\phi ,m} \right)} \!\triangleq\!
\begin{cases}
1,\qquad\, {\text{if}}\  m \!<\! {M_\phi },\\
H_{i,d},\quad\, {\text{if}}\ m = {M_\phi },
\end{cases}
\qquad \forall i,d,\phi,
\end{equation}
where $H_{i,j}$ denotes the shortest distance (in number of hops) from node $i$ to node $j$. 
We note that $Y_i^{(d,\phi,m)}=1$ for all processable commodities because, throughout this paper, we have assumed that every function is available at all cloud network nodes.
In Sec. \ref{sec: extensions_subset}, we discuss a straight-forward generalization of our model, in which each service function is available at a subset of cloud network nodes, in which case, $Y_i^{\left( {d,\phi ,m} \right)}$ for each processable commodity is defined as the shortest distance to the closest node that can process commodity $(d,\phi,m)$.

The enhanced EDCNC-L and EDCNC-Q algorithms work just like their DCNC-L and DCNC-Q counterparts, but using $\hat Q_i^{(d,\phi,m)}(t)$ in place of $Q_i^{(d,\phi,m)}(t)$ to make local processing and transmission scheduling decisions.

\vspace{-0.1cm}
\section{Performance Analysis}
\label{sec: performance_analysis}

In this section, we analyze the performance of the proposed DCNC algorithms.
To facilitate the analysis, we define the following parameters:
\begin{itemize}
\item $A_{\max}$: the constant that bounds the aggregate input rate at all the cloud network nodes; specifically,
$\mathop{\max}\nolimits_{i \in \mathcal V}\mathbb{E}\{ [{\sum\nolimits_{\left( {d,\phi,m} \right)} {a_i^{\left( {d,\phi,m} \right)}\left( t \right)} }]^4 \} \le ({A_{\max }})^4$.

\item $C_{\text{pr}}^{\max}$: the maximum processing capacity among all cloud network nodes; \ie
$C_{\text{pr}}^{\max } \triangleq \mathop {\max }\nolimits_{i \in \mathcal V} \{ {{C_{i,{K_i}}}} \}$.

\item $C_{\text{tr}}^{\max}$: the maximum transmission capacity among all cloud network links; \ie
$C_{\text{tr}}^{\max } \triangleq \mathop {\max }\nolimits_{(i,j) \in \mathcal E} \{ {{C_{ij,{K_{ij}}}}} \}$.

\item $\xi_{\max}$: the maximum flow scaling factor among all service functions; \ie $\xi_{\max}\triangleq\max\nolimits_{(\phi,m)}\{\xi^{(\phi,m)}\}$.

\item $r_{\min}$: the minimum transmission-processing flow ratio among all service functions; \ie  $r_{\min}\triangleq\min\nolimits_{(\phi,m)}\{r^{(\phi,m)}\}$.


\item $\delta_{\max}$: the maximum degree among all cloud network nodes, \ie $\delta_{\max}\triangleq \mathop {\max }\nolimits_{i\in \mathcal V}\{\delta^{\!+\!}(i)+\delta^{\!-\!}(i)\}$.
\end{itemize}

\vspace{-0.3cm}
\subsection{Average Cost and Network Stability}


\begin{thm}
\label{thm: stability}
If the  average input rate matrix $\bm \lambda = (\lambda_i^{(d,\phi,m)})$ is interior to the cloud network capacity region $\Lambda(\mathcal G, \Phi)$,
then the DCNC algorithms stabilize the cloud network, while achieving arbitrarily close to minimum average cost $\overline h^*({\bm \lambda})$ with probability $1$ (w.p.$1$), \ie
\begin{align}
& \!\!{\limsup\limits_{t \rightarrow \infty} \! \frac{1}{t}  \sum\nolimits_{\tau=0}^{t-1}  h(\tau) \leq \overline h^*({\bm \lambda}) +\frac{N B}{V},} \ \ \  (w.p.1)\label{eq_average_cost} \\
& \!\!{\mathop {\lim \sup }\limits_{t \to \infty } \!\frac{1}{t}\!\sum\limits_{\tau  = 0}^{t - 1} \!{\sum\limits_{(d,\phi ,m),i} \!\!\! Q_i^{(d,\phi ,m)}\!(\tau )} \!\le\! \frac{{N\!B \!\!+\!\! V[{{\overline h}^*}\!(\bm \lambda  \!+\! \kappa{\bf{1}}) \!\!-\!\! {{\overline h}^*}\!(\bm \lambda )]}}{\kappa }}\!,\nonumber\\
&\qquad\qquad\qquad\qquad\qquad\qquad\qquad\qquad\qquad\quad (w.p.1)\label{eq_que}
\end{align}
where
\begin{equation}
B =
\begin{cases}
{{B_0},{\text{\ \ \ \ \ under\ DCNC\texttt{-}L\ and\ DCNC\texttt{-}Q,}}}\\
{{B_1},{\text{\ \ \ \ \ under\ EDCNC\texttt{-}L\ and\ EDCNC\texttt{-}Q,}}}
\end{cases}
\end{equation}
with $B_0$ and $B_1$ being positive constants determined by the system parameters $C_{\emph{pr}}^{\max}$, $C_{\emph{tr}}^{\max}$, $A_{\max}$, $\xi_{\max}$, and $r_{\min}$; and $\kappa$ is a positive constant
satisfying $\left({\bm \lambda}+\kappa \bf 1\right)\in \Lambda$. 

\end{thm}

\begin{proof}
The proof of Theorem \ref{thm: stability} is given in Appendix \ref{appendix_stability}.
\end{proof}

Theorem \ref{thm: stability} shows that the proposed DCNC algorithms achieve the average cost-delay
tradeoff $[O(1/V),O(V)]$ with probability 1.\footnote{By setting $\epsilon=1/V$, where $\epsilon$ denotes the deviation from the optimal solution (see Theorem \ref{thm: convergence_rate}), the cost-delay tradeoff is written as $[O(\epsilon), O(1/\epsilon)]$.} Moreover, \eqref{eq_que} holds for any $\bm \lambda$ interior to $\Lambda$, which demonstrates the throughput-optimality of the DCNC algorithms.

\subsection{Convergence Time}

The convergence time of a DCNC algorithm indicates how fast its running time average solution approaches the optimal solution.\footnote{We assume that the local decisions performed by the DCNC algorithms at each timeslot can be accomplished within a reserved computation time within each timeslot, and therefore their different computational complexities are not taking into account for convergence time analysis.}
This criterion is particularly important for online scheduling in settings where the arrival process is non-homogeneous, \ie the average input rate $\bm \lambda$ is time varying. 
In this case, it is important to make sure that the time average solution evolves close enough to the optimal solution
much before the average input rate undergoes significant changes.

We remark that studying the convergence time of a DCNC algorithm involves studying how fast the average cost approaches the optimal value, as well as how fast the flow conservation violation at each node approaches zero.\footnote{Note that the convergence of the flow conservation violation at each node to zero
is equivalent to \emph{strong stability} (see \eqref{eq_que}), if $\bm \lambda$ interior to $\Lambda(\mathcal G,\Phi)$.}

Let $\tilde \mu_{i,\text{pr}}^{(d,\phi,m)}(t)$, $\tilde \mu_{\text{pr},i}^{(d,\phi,m)}(t)$, and $\tilde \mu_{ij}^{(d,\phi,m)}(t)$ denote the {\em actual} flow rates
obtained from removing all \emph{null} packets that may have been assigned when queues do not have enough packets to meet the corresponding assigned flow rates.
Define, for all $i,(d,\phi,m),t$,
\begin{align}
\label{eq_delta_f_def}
\!\! {\Delta f}_i^{(d,\phi ,m)}\!\!\left( t \right) \!\triangleq\! &\sum\limits_{j \in \delta^{\!-\!}\left( i \right)} \!{\tilde \mu _{ji}^{(d,\phi ,m)}\!\!\left( t \right)}  \!+\! \tilde \mu _{\text{pr},i}^{(d,\phi ,m)}\!\!\left( t \right) \!-\! a_i^{\left( {d,\phi ,m} \right)}\!\!\left( t \right)\nonumber\\
&-{\sum\limits_{j \in \delta^{\!+\!}\left( i \right)} {\tilde\mu _{ij}^{(d,\phi ,m)}\!\!\left( t \right)}  - \tilde\mu _{i,\text{pr}}^{(d,\phi ,m)}\!\!\left( t \right)}.
\end{align}
Then, the queuing dynamics is then given by
\begin{equation}
\label{eq_queueing_dynamic2}
Q_i^{\left( {d,\phi ,m} \right)}\left( {t + 1} \right) = Q_i^{\left( {d,\phi ,m} \right)}\left( t \right) +  {\Delta f}_i^{(d,\phi ,m)}\left( t \right).
\end{equation}

The convergence time performance of the proposed DCNC algorithms 
is summarized by the following theorem.
\begin{thm}
\label{thm: convergence_rate}
If the average input rate matrix $\bm \lambda = (\lambda_i^{(d,\phi,m)})$ is interior to the cloud network capacity region $\Lambda(\mathcal G, \Phi)$, then, for all $\epsilon>0$, whenever $t\ge\left.1\right/\epsilon^2$, the mean time average cost and mean time average actual flow rate achieved by the DCNC algorithms during the first $t$ timeslots satisfy: 
\begin{align}
\label{eq_cost_converge_rate}
&\frac{1}{t}\sum\nolimits_{\tau  = 0}^{t - 1} {\mathbb{E}\left\{ {h\left( \tau  \right)} \right\}}  \le {{\overline h}^*}\left( \lambda  \right) + O\left( \epsilon  \right),\\
\label{eq_flow_conserv_converge_rate}
&\frac{1}{t}\sum\nolimits_{\tau  = 0}^{t - 1} {{\mathbb{E}\!\left\{ {\Delta f}_i^{(d,\phi ,m)}\!\!\left( t \right)\right\}}} \le O\!\left( \epsilon  \right),\quad \forall i,(d,\phi,m).
\end{align}
\end{thm}
\begin{proof}
The proof is of Theorem \ref{thm: convergence_rate} given in Appendix \ref{appendix_converge_rate}.
\end{proof}


Theorem \ref{thm: convergence_rate} establishes that, under the DCNC algorithms, both the average cost and the average flow conservation at each node exhibit $O(1/\epsilon^2)$ convergence time to $O(\epsilon)$ deviations from the minimum average cost, and zero, respectively.

\section{Numerical Results}
\label{sec: simulation}
In this section, we evaluate the performance of the proposed DCNC algorithms via numerical simulations in a number of illustrative settings. 
We assume a cloud network based on the continental US Abilene topology shown in Fig. \ref{Abilene}. 
The $14$ cloud network links exhibit homogeneous transmission capacities and costs, while the $7$ cloud network nodes only differ in their processing resource set-up costs. Specifically,
the following two resource settings are considered:

{\em 1) ON/OFF resource levels:} each node and link can either allocate zero capacity, or the maximum available capacity; 
 \ie $K_i=K_{ij}=1$, $\forall i\in \mathcal V, (i,j)\in \mathcal E$. To simplify notation, we define $K\triangleq K_i+1 = K_{ij}+1$, $\forall i\in \mathcal V, (i,j)\in \mathcal E$.
    The processing resource costs and capacities are
    \begin{itemize}
    \item $e_i=1, \forall i\in \mathcal V$; $w_{i,0}=0, \forall i\in \mathcal V$; $w_{i,1}=440, \forall i\in \mathcal V \backslash\{5,6\}$; $w_{5,1}=w_{6,1}=110$.
    \item $C_{i,0}=0, C_{i,1}=440, \forall i\in \mathcal V$.\footnote{The maximum capacity is set to $440$ in order to guarantee that there is no congestion at any part of the network for the service setting considered in the following.}
    \end{itemize}
    The transmission resource costs and capacities are
    \begin{itemize}
    \item $e_{ij}=1, w_{ij,0}=0, w_{ij,1}=440, \forall (i,j)\in \mathcal E$.
    \item $C_{ij,0}=0, C_{ij,1}=440, \forall (i,j)\in \mathcal E$.
    \end{itemize}

{\em 2) Multiple resource levels:} the available capacity at each node and link is split into $10$ resource units; 
\ie $K=11, \forall i\in \mathcal V, (i,j)\in\mathcal E$.
The processing resource costs and capacities are
    \begin{itemize}
    \item $e_{i}=1,\forall i\in\mathcal V$; \\
    $[w_{i,0},w_{i,1},\cdots,w_{i,10},w_{i,11}]\!=\![0,11,\cdots,99,110]$, for $i=5, 6$; \\
    $[w_{i,0},w_{i,1},\cdots,w_{i,10},w_{i,11}]\!=\![0,44,\cdots,396,440], \forall i\in \mathcal V\backslash\{5,6\}$;
    \item $[C_{i,0},C_{i,1},\cdots,C_{i,10},C_{i,11}]\!=\![0,44,\cdots,396,440], \forall i$.
    \end{itemize}
    The transmission resource costs and capacities are
    \begin{itemize}
    \item $e_{ij}=1, \forall (i,j)\in \mathcal E$; \\
    $[w_{ij,0},w_{ij,1},\cdots,w_{ij,10},w_{ij,11}]=[0,44,\cdots,396,440]$, $\forall (i,j)\in \mathcal E$.
    \item $[C_{ij,0},C_{ij,1},\cdots,C_{ij,10},C_{ij,11}]=[0,44,\cdots,396,440]$, $\forall (i,j)\in \mathcal E$.
    \end{itemize}

Note that, for both ON/OFF and multi-level resource settings, the processing resource set-up costs at node $5$ and $6$ are $4$ times cheaper than at the other cloud network nodes.

We consider $2$ service chains, each composed of $2$ virtual network functions: VNF (1,1) (Service $1$, Function $1$) with flow scaling factor $\xi^{(1,1)}=1$; VNF $(1,2)$ with $\xi^{(1,2)}=3$ (expansion function); VNF $(2,1)$ with $\xi^{(2,1)}=0.25$
(compression function); and VNF $(2,2)$ with $\xi^{(2,2)}=1$.
All functions have processing-transmission flow ratio $r^{(\phi,m)}=1$, 
and can be implemented at all cloud network nodes.
Finally, we assume $110$ clients per service, corresponding to all the source-destination pairs in the Abilene network.

\begin{figure}
        \includegraphics[height=3.8cm]{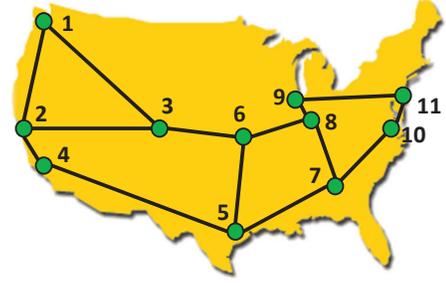}
        \centering{
        \caption{Abilene US Continental Network. Nodes are indexed as: 1) Seattle, 2) Sunnyvale, 3) Denver, 4) Los Angeles, 5) Houston, 6) Kansas City, 7) Atlanta, 8) Indianapolis, 9) Chicago, 10) Washington, 11) New York. }
        \label{Abilene}}
        \vspace{-0.5cm}
\end{figure}

\begin{figure*}[ht]
\centering
\subfigure[]{
\centering \includegraphics[width=5.8cm,]{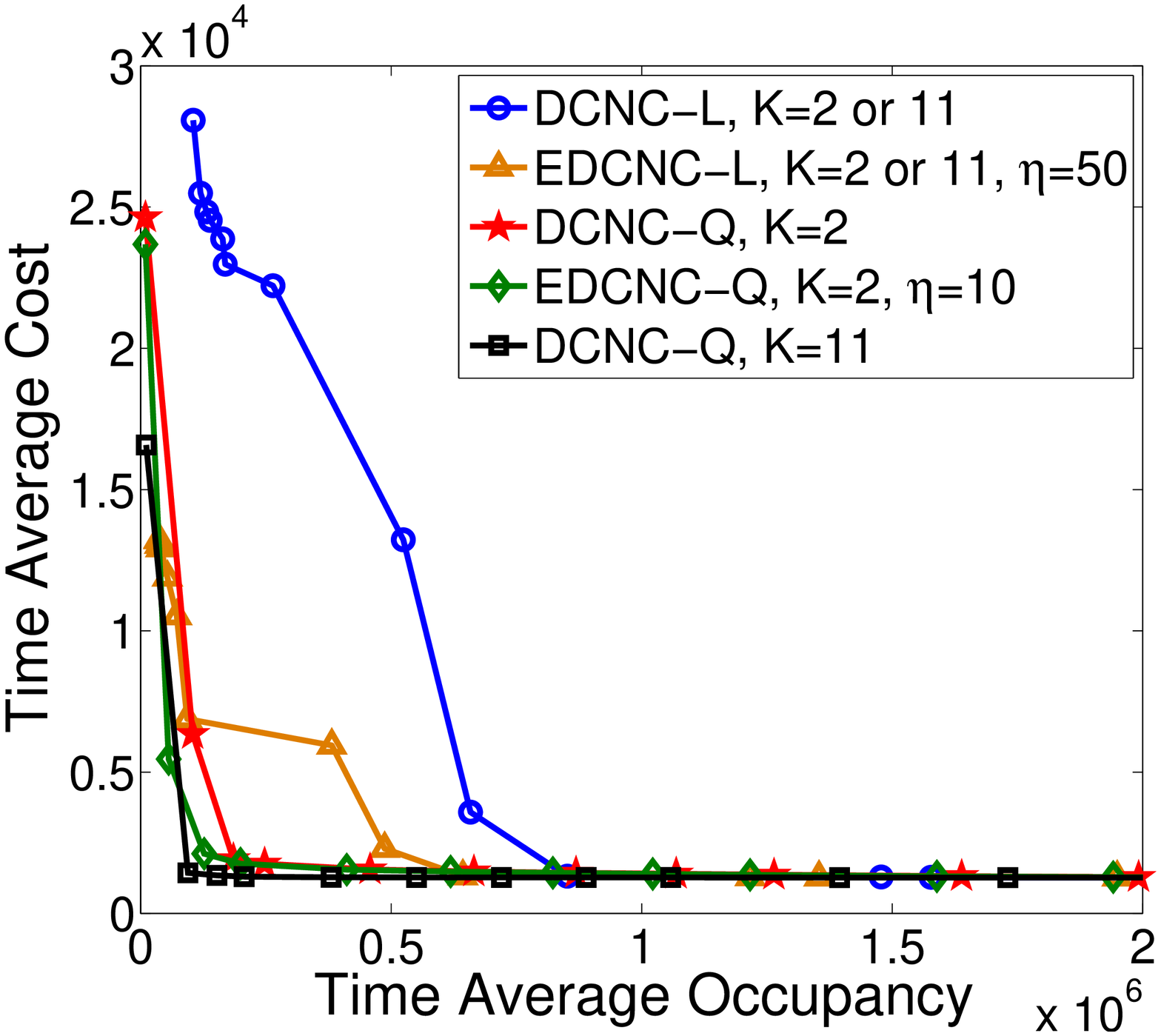}
\label{result1_1}
}
\hspace{-0.7cm}
\subfigure[]{
\centering \includegraphics[width=5.8cm]{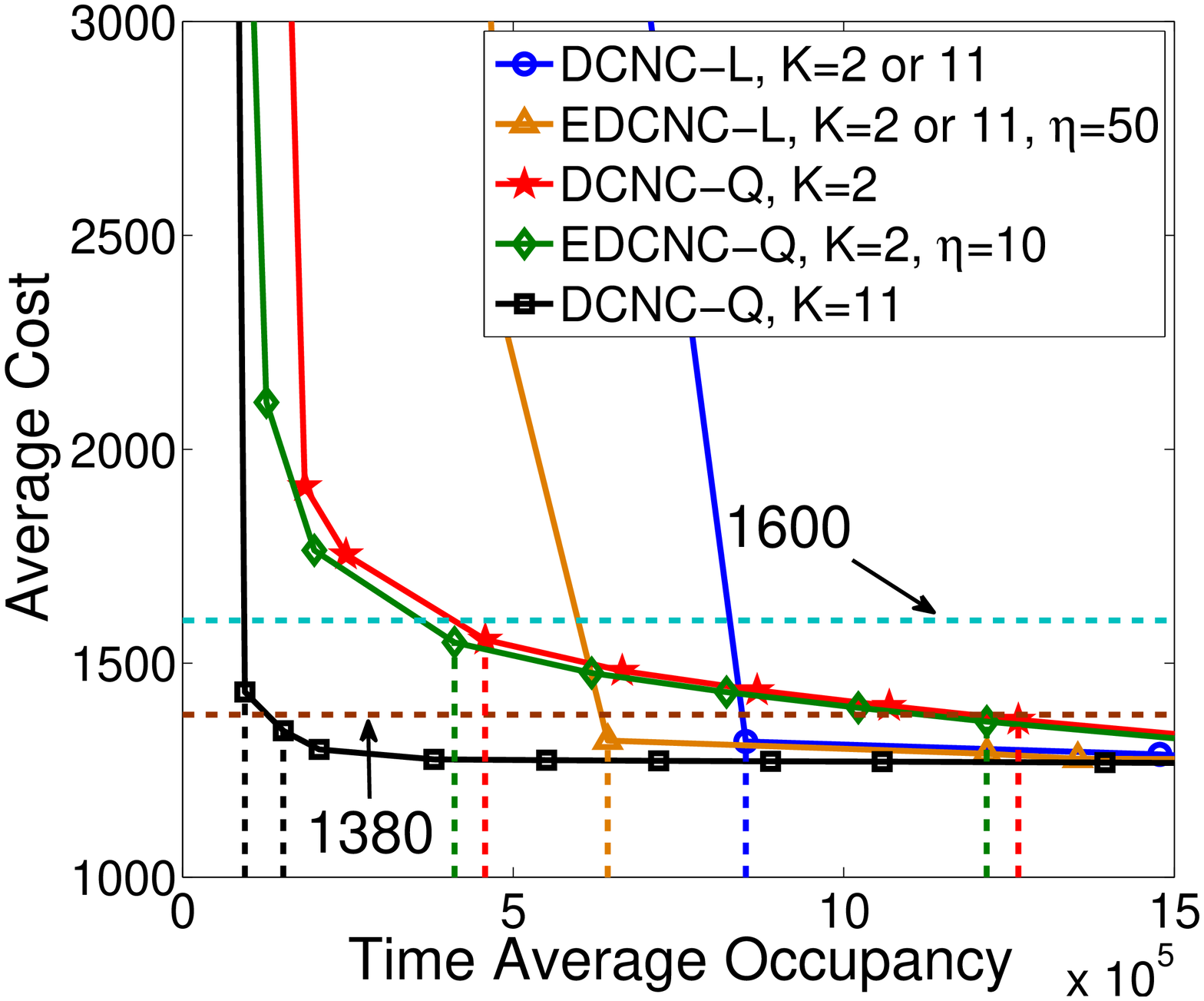}
\label{result1_2}
}
\hspace{-0.7cm}
\subfigure[]{
\centering \includegraphics[width=5.8cm]{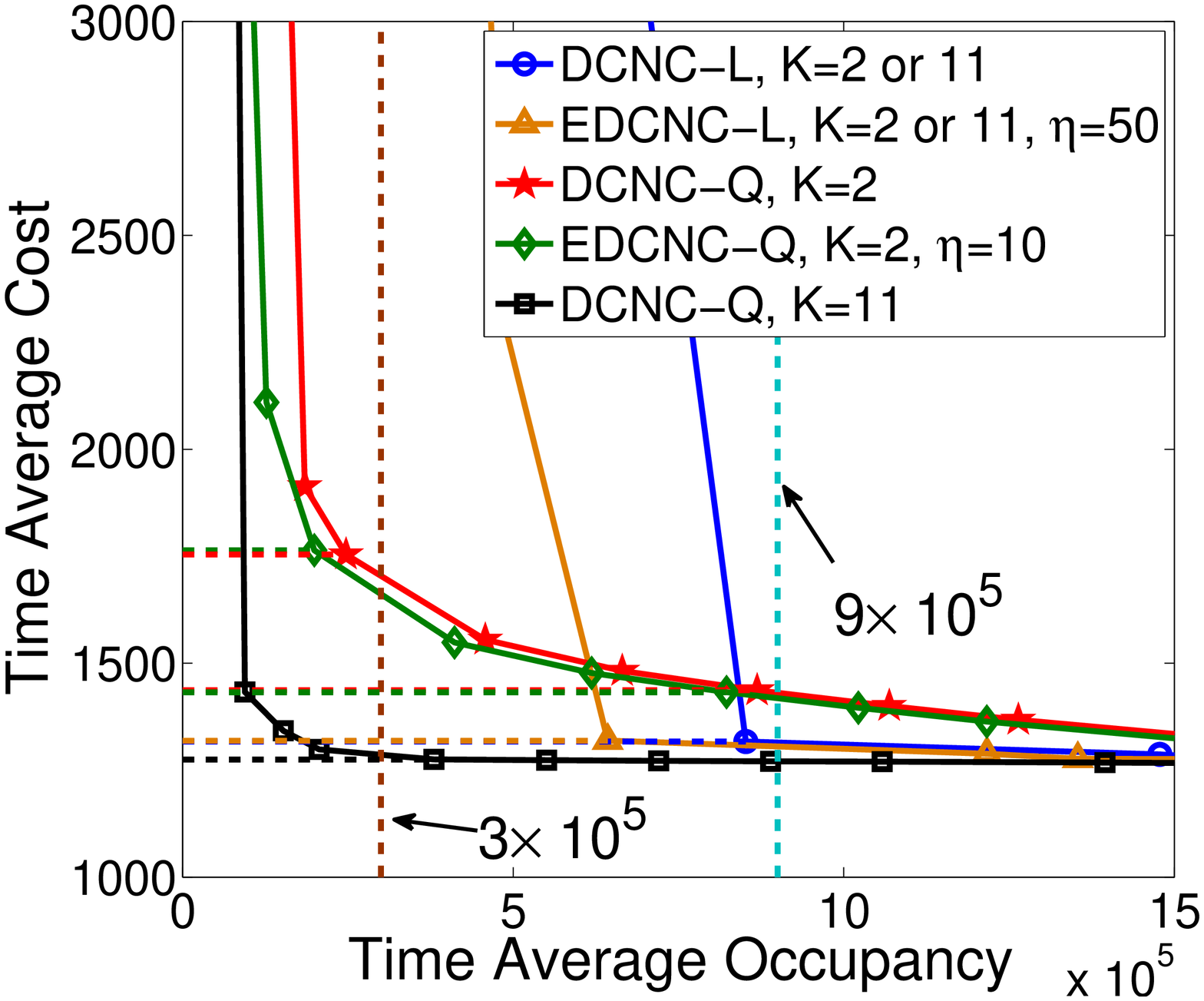}
\label{result1_3}
}
\hspace{-0.7cm}
\subfigure[]{
\centering \includegraphics[width=5.8cm]{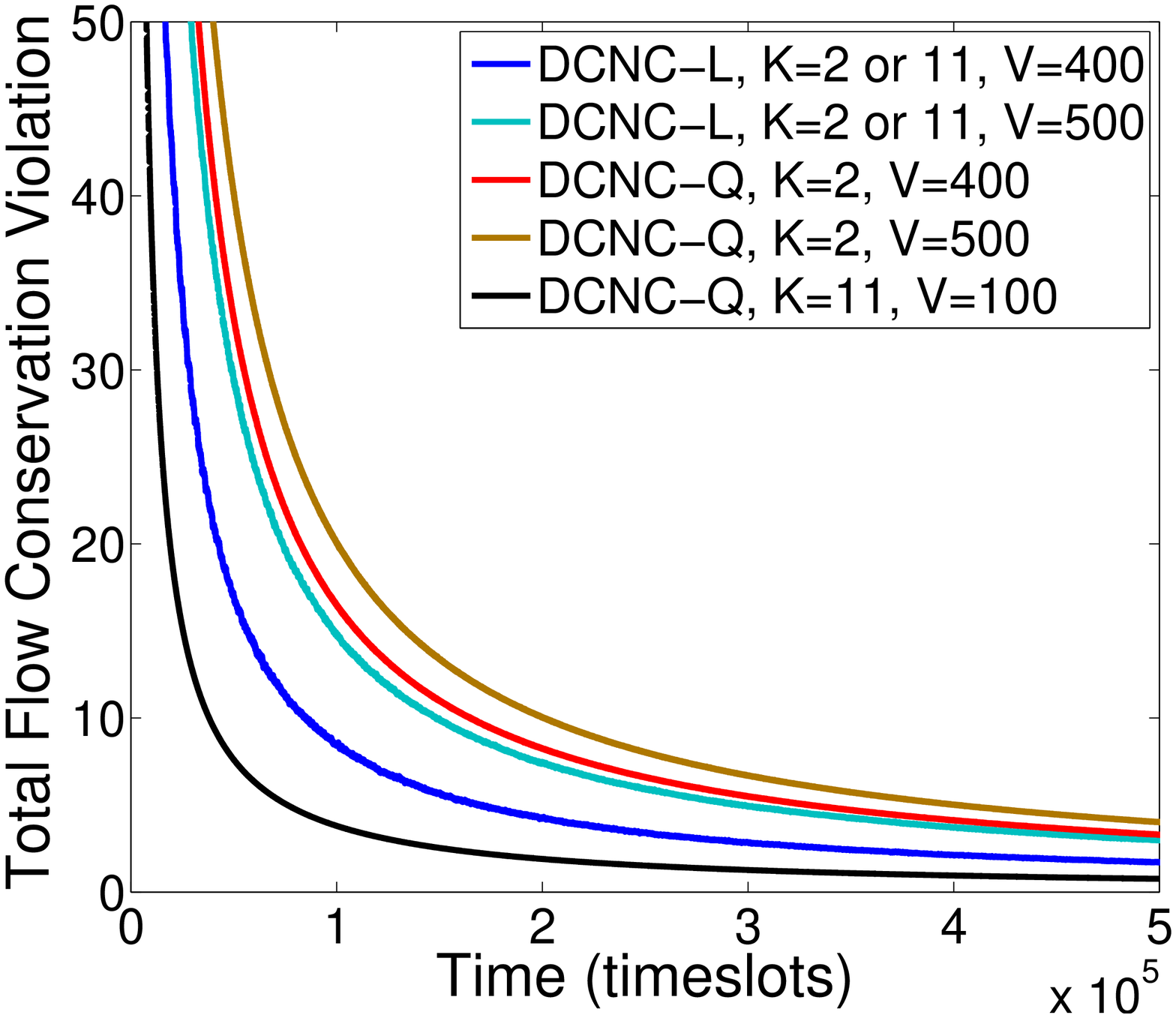}
\label{result1_4}
}
\hspace{-0.7cm}
\subfigure[]{
\centering \includegraphics[width=5.8cm]{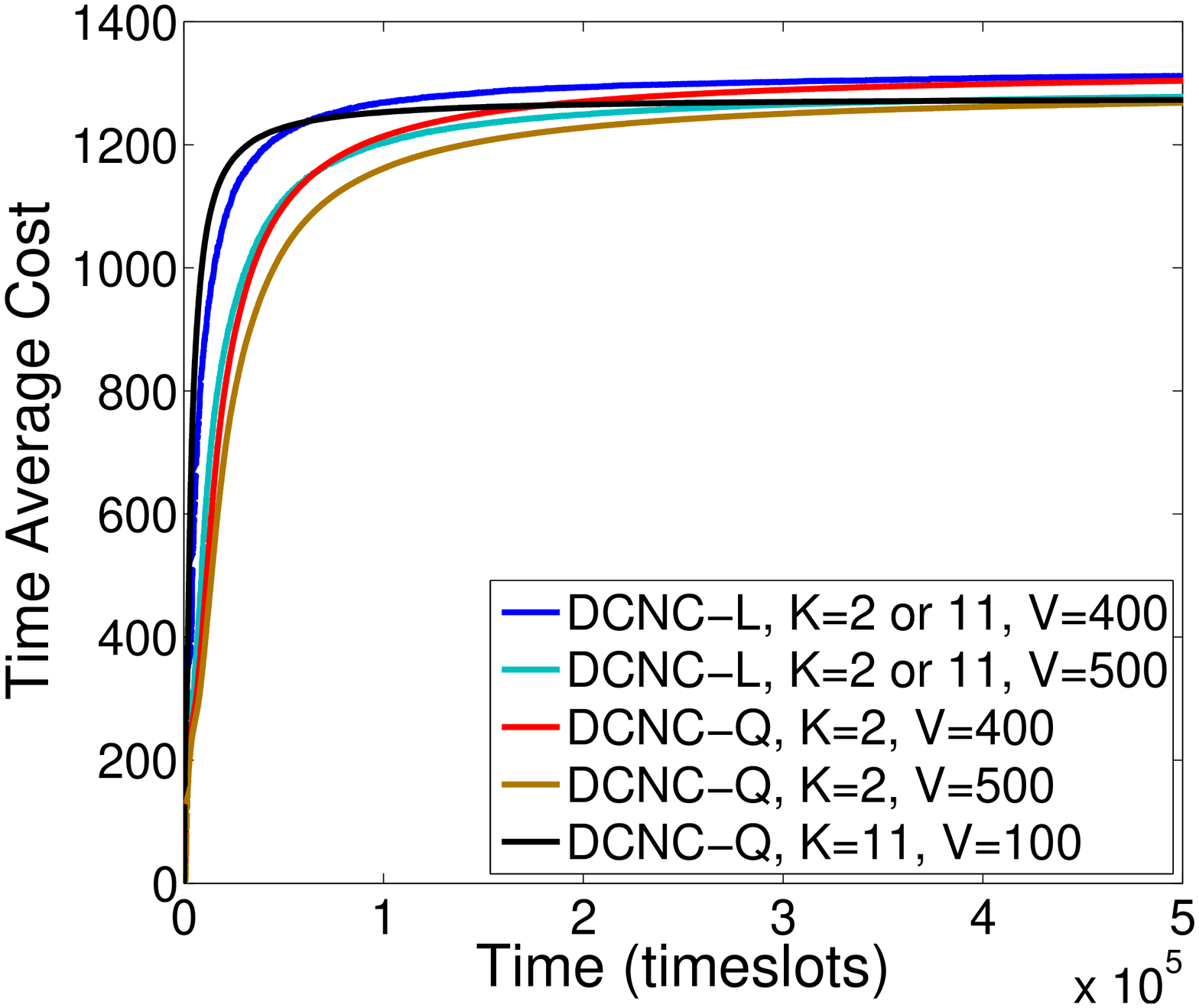}
\label{result1_5}
}
\hspace{-0.7cm}
\subfigure[]{
\centering \includegraphics[width=5.8cm]{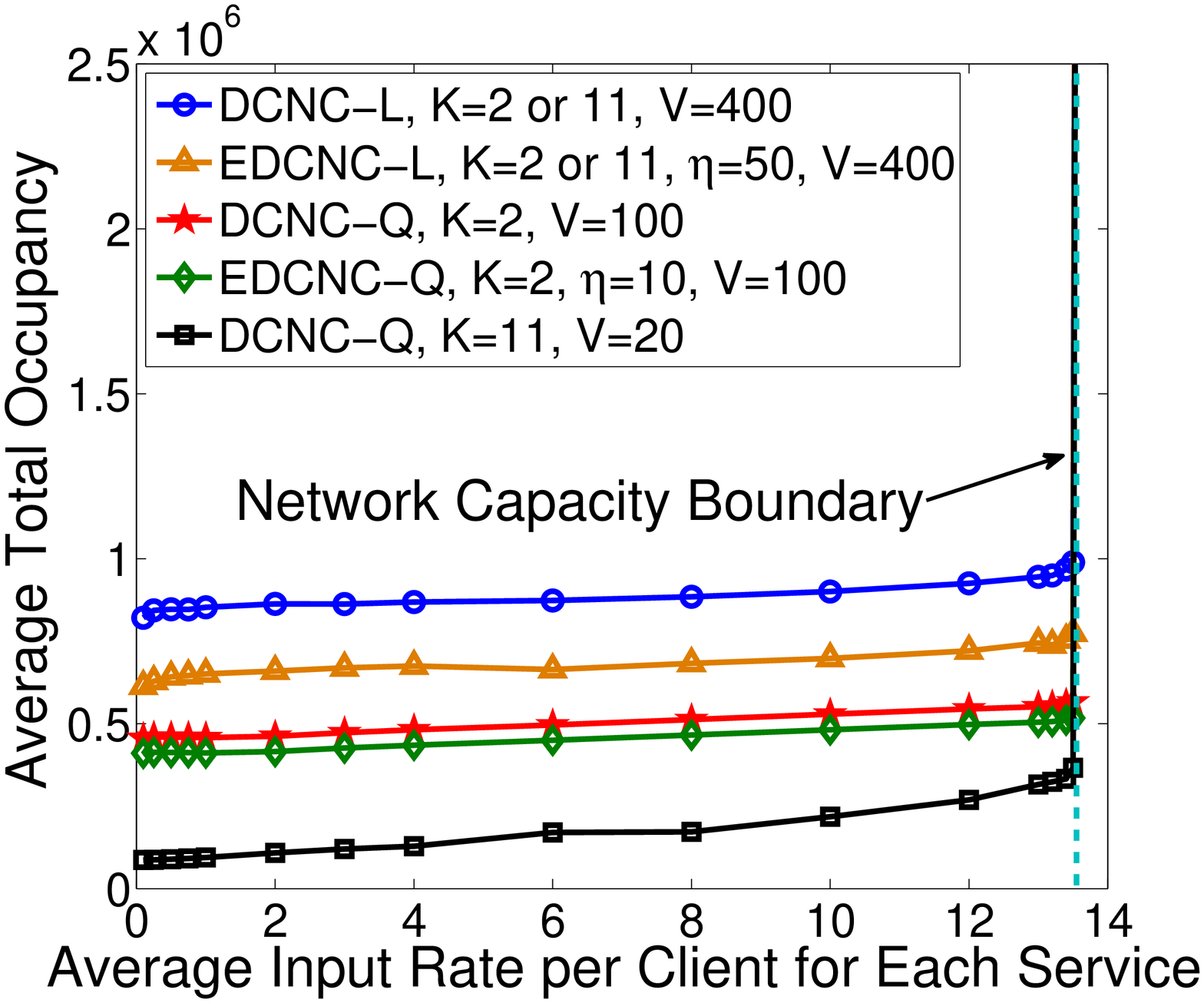}
\label{result1_6}
}
\caption{Performance of DCNC algorithms. a)~Time Average Occupancy v.s. Time Average Cost: a general view; b)~Time average Occupancy v.s. Time Average Cost: given a target average cost c) ~Time average Occupancy v.s. Time Average Cost: given a target average occupancy;  d)~Total flow conservation violation evolution over time: effect of the $V$ value; e) ~Time average cost evolution over time: effect of the $V$ value; f)~Time Average Occupancies with varying service input rate: throughput optimality}
\vspace{-0.6cm}
\label{fig: result 1}
\end{figure*}


\vspace{-0.3cm}
\subsection{Cost-Delay Tradeoff}
\label{subsec: simu_cost_vs_delay}

Figs. \ref{result1_1}-\ref{result1_3} show the tradeoff between the time average cost and the time average end-to-end delay (represented by the total time average occupancy or queue backlog), under the different DCNC algorithms. The input rate of all source commodities is set to $1$ and the the cost/delay values are obtained after simulating each algorithm for $10^6$ timeslots. 
Each tradeoff curve is obtained by varying the control parameter $V$ between $0$ and $1000$ for each algorithm. Small values of $V$ favor low delay at the expense of high cost, while large values of $V$ lead to points in the tradeoff curves with lower cost and higher delay.

It is important to note that since the two resource settings considered, \ie ON/OFF ($K=2$) vs. multi-level ($K=11$), are characterized by the same maximum capacity and the same constant ratios $\left.C_{i,k}\right/w_{i,k}$ 
and $\left.C_{ij,k}\right/w_{ij,k}$, the performance of the linear DCNC algorithms (DCNC-L and EDCNC-L) does not change under the two resource settings. 
On the other hand, the quadratic algorithms (DCNC-Q and EDCNC-Q) can exploit the finer resource granularity of the multi-level resource setting to improve the cost-delay tradeoff.
We also note that for the enhanced versions of the algorithms that use the STPD bias (EDCNC-L and EDCNC-Q), we choose the bias coefficient $\eta$ among the values of multiples of $10$ that leads to the best performance for each algorithm.\footnote{Simulation results for different values of $\eta$ can be found in 
\cite{icc_2016}.}

\vspace{-0.2cm}
Fig. \ref{result1_1} shows how the average cost under all DCNC algorithms reduces at the expense of network delay, and converges to the same minimum value. 
While all the tradeoff curves follow the same $[O(1/V),O(V)]$ relationship established in Theorem \ref{thm: stability},
the specific trade-off ratios can be significantly different.
The general trends observed in Fig. \ref{result1_1} are as follows. DCNC-L exhibits the worst cost-delay tradeoff. Recall that DCNC-L assigns either zero or full capacity to a single commodity in each timeslot, and hence the finer resource granularity of $K=11$ does not improve its performance. However, adding the SDTP bias results in a substantial performance improvement, as shown by the EDCNC-L curve. Now let's focus on the quadratic algorithms. DCNC-Q with $K=2$ further improves the cost delay-tradeoff, at the expense of increased computational complexity. In this case, adding the SDTP bias provides a much smaller improvement (see EDCNC-Q curve), showing the advantage of the more ``balanced" scheduling decisions of DCNC-Q. Finally, DCNC-Q with $K=11$ exhibits the best cost-delay tradeoff, illustrating the ability of DCNC-Q to exploit the finer resource granularity to make ``smoother" resource allocation decisions. In this setting, adding the SDTP bias does not provide further improvement and it is not shown in the figure.

\vspace{-0.05cm}
While Fig. \ref{result1_1} illustrates the general trends in  improvements obtained using the quadratic metric and the SDTP bias, there are regimes in which the lower complexity DCNC-L and EDCNC-L algorithms can outperform their quadratic counterparts. We illustrate these regimes In Figs. \ref{result1_2} and \ref{result1_3}, by zooming into the lower left of Fig. \ref{result1_1}. 
As  shown in Fig. \ref{result1_2}, for the case of $K=2$, the cost-delay curves of DCNC-L and EDCNC-L cross with the curves of DCNC-Q and EDCNC-Q.
For example, for a target cost of $1380$, 
DCNC-L and EDCNC-L result in lower average occupancies ($8.52\times10^5$ and $6.43\times10^5$) than DCNC-Q ($1.26\times10^6$) and EDCNC-Q ($1.21\times10^6$).
On the other hand, if we increase the target cost to $1600$, DCNC-Q and
EDCNC-Q achieve lower occupancy values ($4.58\times10^5$ and $4.11\times10^5$) than DCNC-L ($8.52\times10^5$) and EDCNC-L ($6.43\times10^5$).
Hence, depending on the cost budget, there may be a regime in which the simpler DCNC-L and EDCNC-L algorithms become a better choice. However, this regime does 
not exist for $K=11$, where the average occupancies under DCNC-Q ($1342$ and $1433$ respectively for the two target costs) are much lower than (E)DCNC-L.
In Fig. \ref{result1_3}, we compare cost values for given target occupancies. With $K=2$ and a target average occupancy of $9\times10^5$, the average costs achieved by DCNC-L ($1317$) and EDCNC-L ($1319$) are lower than those achieved by DCNC-Q ($1437$) and EDCNC-Q ($1432$). In contrast, if we reduce the target occupancy to  $3\times10^5$, DCNC-Q and EDCNC-Q (achieving average costs $1754$ and $1764$) outperform DCNC-L and EDCNC-L (with cost values $2.64\times10^4$ and $6879$ beyond the scope of Fig. \ref{result1_3}).
With $K=11$, DCNC-Q achieves average costs of $1286$ and $1271$ for the two target occupancies, outperforming all other algorithms.

\begin{figure*}[ht]
\centering
\begin{minipage}[t]{0.1\textwidth}
\centering \includegraphics[height=3.5cm,]{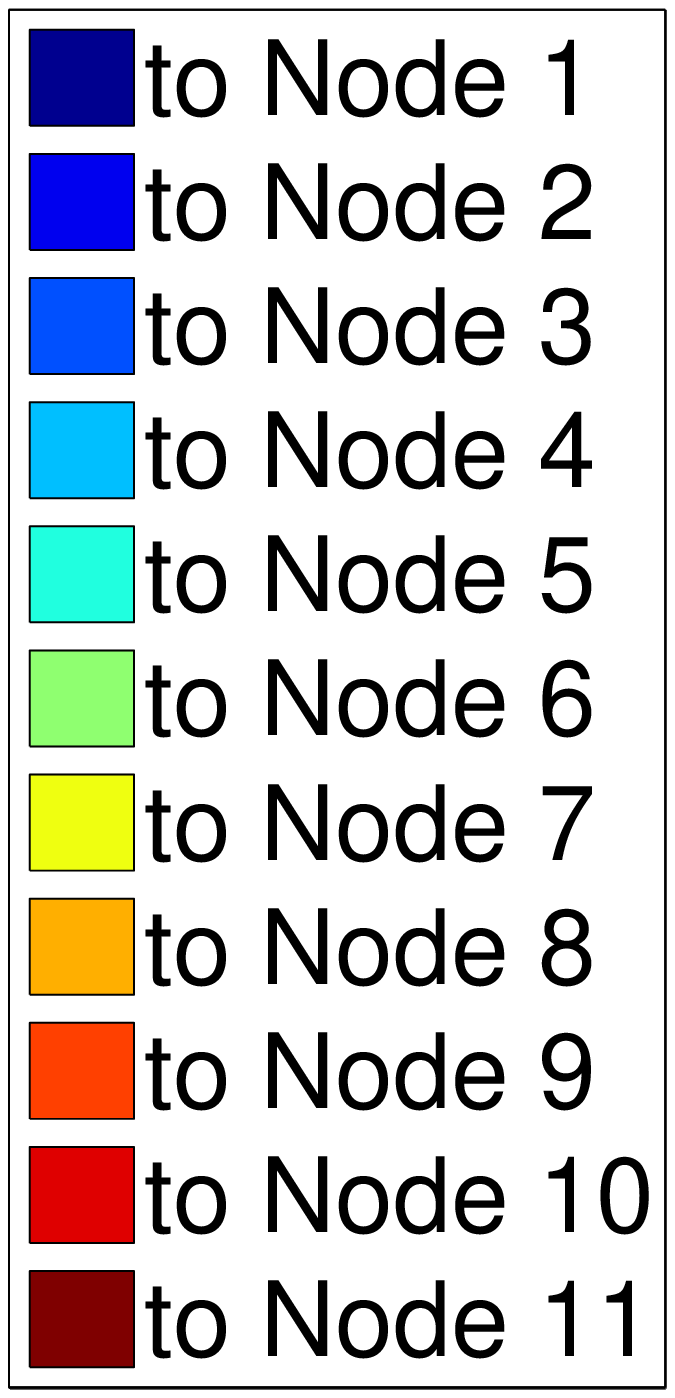}
\end{minipage}
\subfigure[]{
\centering \includegraphics[height=3.4cm,]{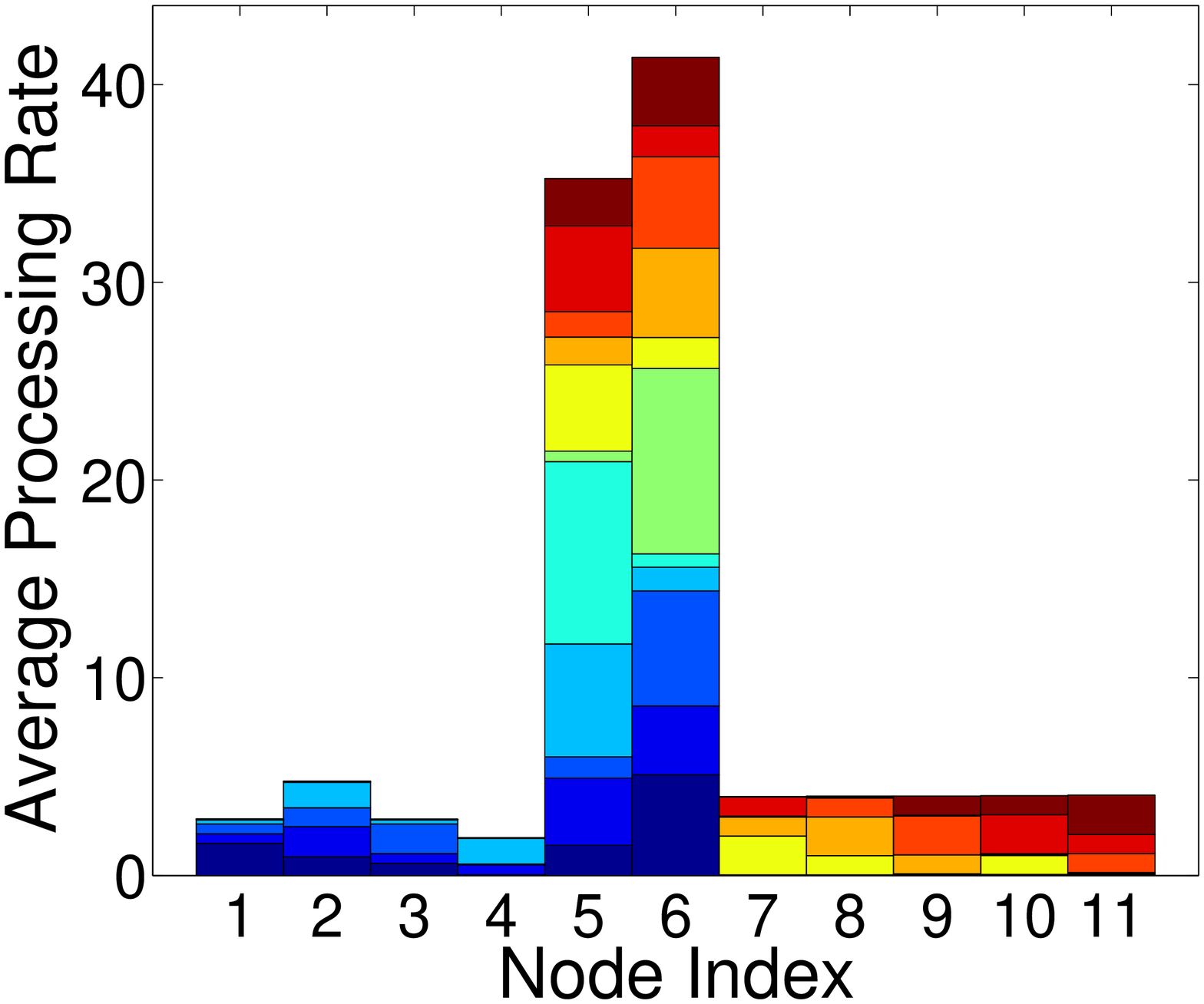}
\label{result2_1}
}
\hspace{-0.8cm}
\subfigure[]{
\centering \includegraphics[height=3.4cm]{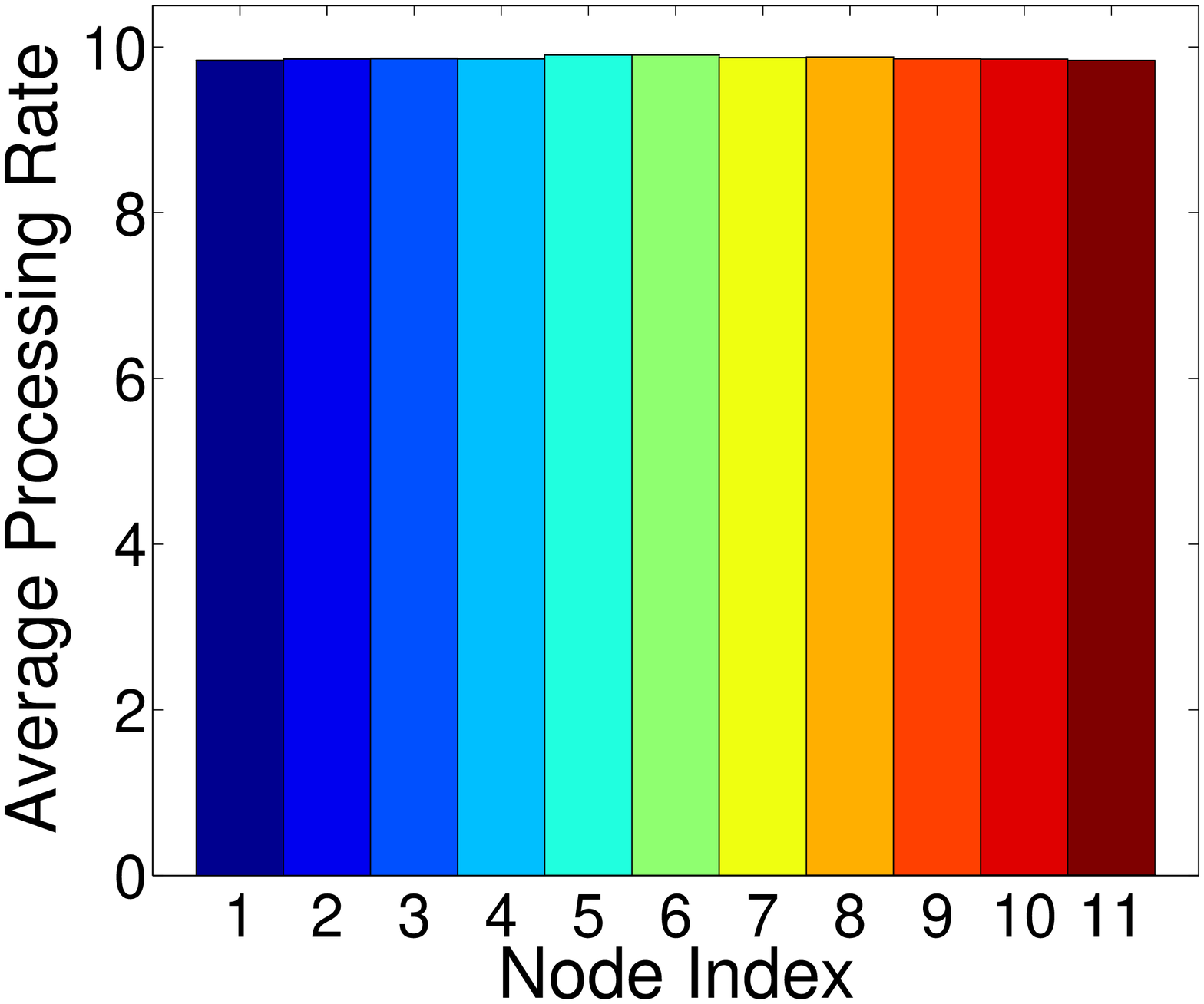}
\label{result2_2}
}
\hspace{-0.8cm}
\subfigure[]{
\centering \includegraphics[height=3.4cm]{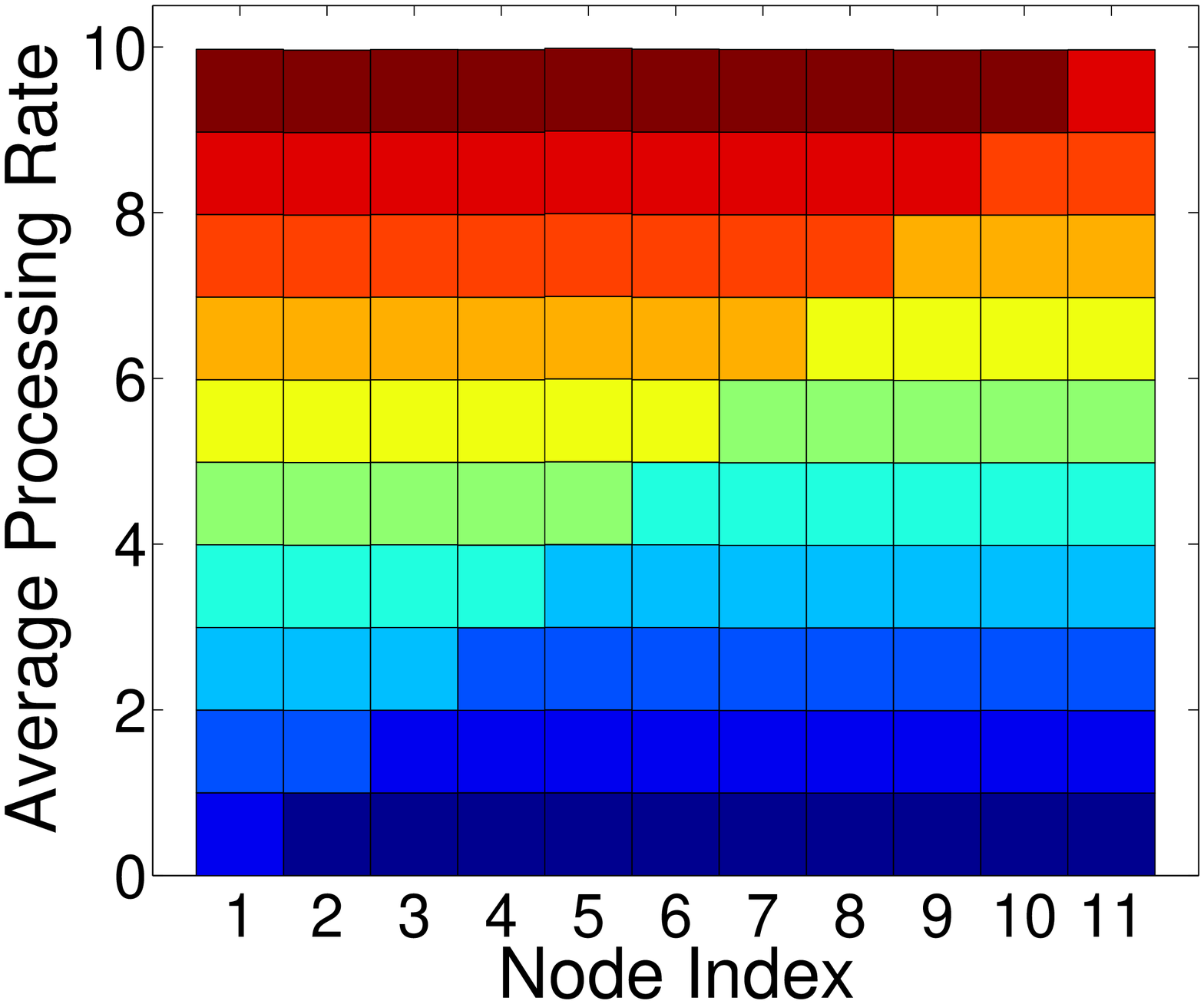}
\label{result2_3}
}
\hspace{-0.8cm}
\subfigure[]{
\centering \includegraphics[height=3.4cm]{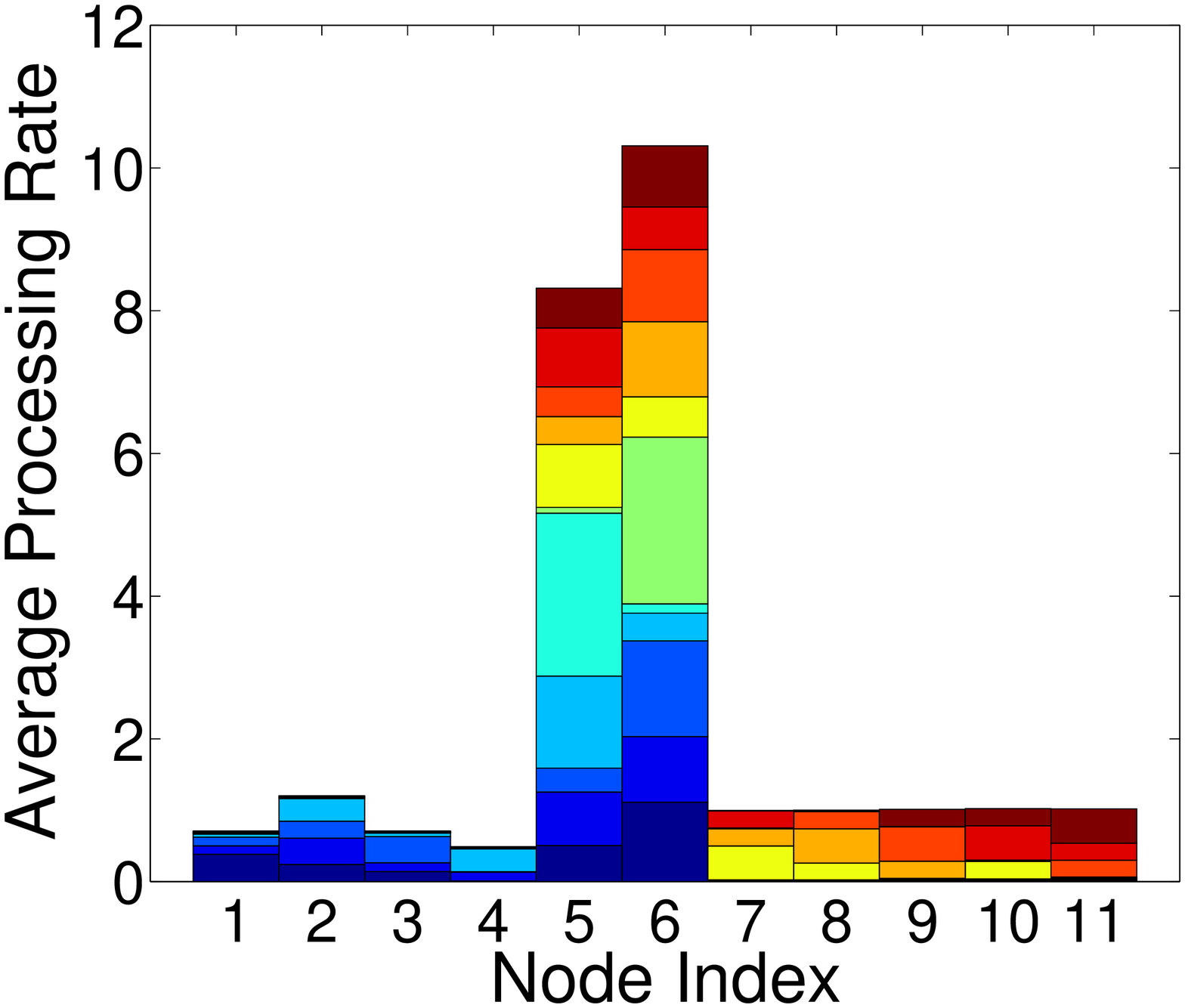}
\label{result2_4}
}
\caption{Average Processing Flow Rate Distribution. a)~Service 1, Function 1; b)~Service 1, Function 2; c)~Service 2, Function 1; d)~Service 2, Function 2. }
\vspace{-0.5cm}
\label{fig: result 2}
\end{figure*}

\vspace{-0.2cm}
\subsection{Convergence Time}

\label{subsec: simu_converge}

In Figs. \ref{result1_4} and \ref{result1_5}, we show the time evolution of the total flow conservation violation (obtained by summing over all nodes and commodities, 
the absolute value of the flow conservation violation) and the total time average cost, respectively. The average input rate of each source commodity is again set to $1$.
As expected, observe how decreasing the value of $V$ speeds up the convergence of all DCNC algorithms.
However, note from Fig. \ref{result1_5} that the converged time average cost is higher with a smaller value of $V$, 
consistent with the tradeoff established 
in Theorem \ref{thm: stability}.
Note that the slower convergence of DCNC-Q with respect to DCNC-L with the same value of $V$
does not necessarily imply a disadvantage of DCNC-Q.
In fact, due to its more balanced scheduling decisions, DCNC-Q can be designed with a smaller $V$ than DCNC-L, in order to enhance convergence speed while achieving no worse cost/delay performance.
This effect is obvious in the case of $K=11$.
As shown in Fig. \ref{result1_4} and Fig. \ref{result1_5}, with $K=11$, DCNC-Q with $V=100$ achieves faster convergence than DCNC-L with $V=400$, while their converged average cost values are similar.




\subsection{Capacity Region}
\label{subsec: simu_throughput}

Fig. \ref{result1_6} illustrates the throughput performance of the DCNC algorithms by showing the time average occupancy as a function to the input rate (kept the same for all source commodities). The simulation time is $10^6$ timeslots and
the values of $V$ used for each algorithm are chosen according to Fig. \ref{result1_2} in order to guarantee that the average cost is lower than the target value $1600$.
As the average input rate increases to $13.5$, the average occupancy under all the DCNC algorithms exhibits a sharp raise, illustrating the boundary of the cloud network capacity region (see \eqref{eq_que} and let $\kappa\rightarrow 0$).

Observe, once more, the improved delay performance achieved via the use of the STPD bias and the quadratic metric in the proposed control algorithms.


\subsection{Processing Distribution}
\label{subsec: simu_processing_distribution}

Fig. \ref{fig: result 2} shows the optimal average processing rate distribution across the cloud network nodes for each service function under the ON/OFF resource setting ($K=2$). We obtain this solution, for example, by running DCNC-L with $V=1000$ for $10^6$ timeslots.
The processing rate of function $(\phi,m)$ refers to the processing rate of its input commodity $(d,\phi,m-1)$.

Observe how the implementation of VNF $(1,1)$ mostly concentrates at node $5$ and $6$, which are the cheapest processing locations 
However, note that part of VNF $(1,1)$ for destinations in the west coast (nodes $1$ through $4$) takes place at the west coast nodes, illustrating the fact that while processing is cheaper at nodes $5$ and $6$, shorter routes can compensate the extra processing cost at the more expensive nodes. 
A similar effect can be observed for destinations in the east coast, where part of VNF$(1,1)$ takes place at east coast nodes.

Fig. \ref{result2_2} shows the average processing rate distribution for VNF $(1,2)$. Note that VNF $(1,2)$ is an expansion function. This results in the processing of commodity $(d,1,1)$ completely concentrating at the destination nodes, in order to minimize the impact of the extra cost incurred by the transmission of the larger-size commodity $(d,1,2)$ resulted from the execution of VNF $(1,2)$.

For Service $2$, note that VNF $(2,1)$ is a compression function. As expected, the implementation of VNF $(2,1)$ takes place at the source nodes, in order to reduce the transmission cost of Service $2$ by compressing commodity $(d,2,0)$ into the smaller-size commodity $(d,2,1)$ even before commodity $(d,2,0)$ flows into the network. As a result, as shown in Fig. \ref{result2_3}, for all $1\le d\le 11$, commodity $(d,2,0)$ is processed at all the nodes except node $d$, and the average processing rate of commodity $(d,2,0)$ at each node $i\neq d$ is equal to $1$, which is the average input rate per client.

Fig. \ref{result2_4} shows the average processing rate distribution for VNF $(2,2)$, which exhibits a similar distribution to VNF $(1,1)$, except for having different rate values, due to the compression effect of VNF $(2,1)$. 

\section{Extensions}
\label{sec: extensions}

In this section, we discuss interesting extensions of the DCNC algorithms presented in this paper that can be easily captured via simple modifications to our model.

\subsection{Location-Dependent Service Functions}
\label{sec: extensions_subset}

For ease of notation and presentation, throughout this paper, we have implicitly assumed that every cloud network node can implement all network functions. 
In practice, each cloud network node may only host a subset of functions $\widetilde {\mathcal M}_{\phi,i}\subseteq \mathcal M_\phi, \forall \phi\in \Phi$.
In this case, the local processing decisions at each node would be made by considering only those commodities that can be processed by the locally available functions. 
In addition, the STPD bias $Y_i^{(d,\phi,m)}$ would need to be updated as, for all $i,d,\phi$,
\begin{equation}
Y_i^{\left( {d,\phi ,m} \right)} \!\triangleq\!
\begin{cases}
\mathop {\min }\limits_{j: j\in \mathcal V, (m+1) \in \widetilde{\mathcal M}_{\phi,j}} \left\{ {H_{i,j}+1}  \right\},\ \   {\text{if}}\  m \!<\! {M_\phi },\\
H_{i,d},\quad {\text{if}}\ m = {M_\phi }.\notag
\end{cases}
\end{equation}

\subsection{Propagation Delay}

In this work, we have assumed that network delay is dominated by queueing delay, and ignored propagation delay.
However, in large-scale cloud networks, where communication links can have large distances, the propagation of data across two neighbor nodes may incur non-negligible delays.
In addition, while much smaller, the propagation delay incurred when forwarding packets for processing in a large data center may also be  non-negligible.
In order to capture propagation delays, let
$D_i^{\text{pg}}$ and $D_{ij}^{\text{pg}}$ denote the propagation delay (in timeslots) for reaching the processing unit at node $i$ and for reaching neighbor $j$ from node $i$, respectively. We then have the following queuing dynamics and service chaining constraints:
\begin{eqnarray}
&\hspace{-0.3cm}\Scale[0.97]{ Q_i^{\!(d,\phi,m)}\!(t\!+\!1) \! \leq\!\! \left[\! Q_i^{\!(d,\phi,m)}\!(t) \!- \!\! \displaystyle{\sum_{j\in\delta^{\!+\!}(i)}} \mu_{ij}^{\!(d,\phi,m)}\!(t) \!-\! \mu_{i,\text{pr}}^{\!(d,\phi,m)}\!(t) \! \right]^{\!\!+}} \notag \\
&  \Scale[0.97]{\ + \displaystyle{\sum_{j\in\delta^{\!-\!}(i)}} \mu_{ji}^{\!(d,\phi,m)}\!(t\!-\!D_{ji}^{\text{pg}}) + \mu_{\text{pr},i}^{\!(d,\phi,m)}\!(t) } + a_i^{\!(d,\phi,m)}\!(t), \label{eq_queueing_dynamic_extend}\\
&\hspace{-2.5cm}\Scale[0.97]{\mu_{\text{pr},i}^{(d,\phi,m)}(t) = \xi^{(\phi,m)}\mu_{i,\text{pr}}^{(d,\phi,m-1)}(t-D_i^{\text{pg}}).}
\label{eq_chain_extend}
\end{eqnarray}

Moreover, due to propagation delay, queue backlog observations become outdated. Specifically, the queue backlog of commodity $(d,\phi,m)$ at node $j\in \delta(i)$ observed by node $i$ at time $t$ is $Q_j^{(d,\phi,m)}(t-D_{ji}^{\text{pg}})$.

Furthermore, for EDCNC-L and EDCNC-Q, the STPD bias $Y_i^{(d,\phi,m)}$, for all $i,d,\phi$, would be updated as
\begin{equation}
Y_i^{\left( {d,\phi ,m} \right)} \!\triangleq\!
\begin{cases}
\mathop {\min }\limits_{j\in \mathcal V} \left\{{\tilde H_{i,j} + D_i^{\text{pg}}}  \right\},\ \   {\text{if}}\  m \!<\! {M_\phi }.\\
\tilde H_{i,d},\quad {\text{if}}\ m = {M_\phi },\notag
\end{cases}
\end{equation}
where $\tilde H_{i,j}$ is the length of the shortest path from node $i$ to node $j$, with link $(u,v)\in \mathcal E$ having length $D_{uv}^{\text{pg}}$.

With \eqref{eq_queueing_dynamic_extend}, \eqref{eq_chain_extend}, and the outdated backlog state observations, the proposed DCNC algorithms can still be applied and be proven to retain the established throughput, average cost, and convergence performance guarantees, while suffering from increased average delay.


\subsection{Service Tree Structure}

While most of today's network services can be described via a chain of network functions, next generation digital services may contain functions with multiple inputs.
Such services can be described via a \emph{service tree},
as shown in Fig. \ref{service_tree}.

In order to capture these type of services, we let
$\mathcal I(\phi,m)$ denote the set of commodities 
that act as input to function $(\phi,m)$, generating commodity $(d,\phi,m)$. 
The service chaining constraints are then updated as
\begin{align}
\mu_{\text{pr},i}^{(d,\phi,m)}\!(t) \!=\! \xi^{(\phi,n)}\!\mu_{i,\text{pr}}^{(d,\phi,n)}\!(t),\ \ \ \forall t,i,d,\phi,m,n\in \mathcal I(\phi,m). \notag
\end{align}
where $\xi^{(\phi,n)}, \forall n\in\mathcal I(\phi,m)$ denotes the flow size ratio between the output commodity $(d,\phi,m)$ and each of its input commodities $n\in\mathcal I(\phi,m)$.
In addition, the processing capacity constraints are updated as 
\begin{align}
&\sum_{(d,\phi,n)}\!\mu_{i,\text{pr}}^{(d,\phi,n)}\!(t)r^{(\phi,n)}\!\le\! \sum_{k\in \mathcal K_i}\!\!{C_{i,k}y_{i,k}(t)}, \quad \forall t,i, \notag
\end{align}
where $r^{(\phi,n)}$ now denotes the computation requirement of processing a unit flow of commodity $(d,\phi,n)$.

Using the above updated constraints in the LDP bound minimizations performed by the DCNC algorithms, we can provide analogous throughput, cost, and convergence time guarantees for the dynamic control of service trees in cloud networks.

\begin{figure}
\centering
\includegraphics[width=8.8cm]{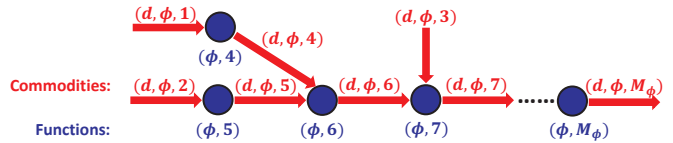}
\caption{A network service tree $\phi\in\Phi$. VNF $(\phi,m)$ takes input commodities $(d,\phi,n)$, $n\in \mathcal I(\phi,m)$, and generates commodity $(d,\phi,m)$.}
\label{service_tree}
\vspace{-0.5cm}
\end{figure}

\section{Conclusions}
\label{sec: conclusions}

We addressed the problem of dynamic control of network service chains in distributed cloud networks, in which demands are unknown and time varying.
For a given set of services, we characterized the cloud network capacity region and
designed online dynamic control algorithms
that jointly schedule flow processing and transmission decisions, along with the corresponding allocation of network and cloud resources. The proposed algorithms stabilize the underling cloud network queuing system, as long as the average input rates are within the cloud network capacity region.
The achieved average cloud network costs can be pushed arbitrarily close to minimum with probability 1, while trading off average network delay.  
Our algorithms converge to within $O(\epsilon)$ of the optimal solutions in time $O(1/\epsilon^2)$. 
DCNC-L makes local transmission and processing decisions with linear complexity with respect to the number of commodities and resource allocation choices. In comparison, DCNC-Q makes local decisions by minimizing a quadratic metric obtained from an upper bound expression of the LDP function, and we show via simulations that the cost-delay tradeoff can be significantly improved. Furthermore, both DCNC-L and DCNC-Q are enhanced by introducing a STPD bias into the scheduling decisions, yielding the EDCNC-L and EDCNC-Q algorithms, which exhibit further improved delay performance.


%

\appendices

\section{Proof of Theorem \ref{thm: network_capacity_region}}
\label{Proof of capacity_region_nessesary}

We prove Theorem \ref{thm: network_capacity_region} by separately proving necessary and sufficient conditions.

\subsection{Proof of Necessity}

We prove that constraints \eqref{eq_thm1_stability}-\eqref{betas} are required for cloud network stability and that $\overline h^*$ given in \eqref{eq_thm1_minimum_cost} is the minimum achievable cost by any stabilizing policy.

Consider an input rate matrix ${\bm \lambda}\in \Lambda(\mathcal G,{\Phi})$.
Then, there exists a stabilizing policy that supports $\bm \lambda$.
We define the following quantities for this stabilizing policy:
\begin{itemize}
\item $X_i^{(d,\phi,m)}(t)$: the number of packets of commodity $(d,\phi ,m)$ exogenously arriving at node $i$, that got delivered within the first $t$ timeslots
\item  $F_{i,\text{pr}}^{\left( {d,\phi ,m} \right)}(t)$ and $ F_{\text{pr},i}^{\left( {d,\phi ,m} \right)}(t)$: the number of packets of commodity $(d,\phi ,m)$ input to and output from the processing unit of node $i$, that got delivered within the first $t$ timeslots, respectively;
\item $F_{ij}^{\left( {d,\phi ,m} \right)}(t)$: the number of packets of commodity $(d,\phi ,m)$ transmitted through link $(i,j)$, that got delivered within the first $t$ timeslots
\end{itemize}
where we say that a packet of commodity $(d,\phi,m)$ got {\em delivered} within the first $t$ timeslots, if it got processed by functions $\{(\phi,m+1),\dots,(\phi,M_{\phi})\}$ and the resulting packet of the final commodity $(d,\phi,M_{\phi})$ exited the network at destination $d$ within the first $t$ timeslots.

The above quantities satisfy the following conservation law: 
\begin{align}
\label{eq_thm1_stabilityproof}
&\sum\nolimits_{j \in \delta^{\!-\!} \left( i \right)} {F_{ji}^{\left( {d,\phi ,m} \right)}}(t)  + F_{\text{pr},i}^{\left( {d,\phi ,m} \right)}(t)+ X_i^{(d,\phi,m)}(t)
= \nonumber\\
&\sum\nolimits_{j \in \delta^{\!+\!} \left( i \right)} {F_{ij}^{\left( {d,\phi ,m} \right)}}(t)  + F_{i,\text{pr}}^{\left( {d,\phi ,m} \right)}(t),
\end{align}
for all nodes and commodities, except for the final commodities at their respective destinations.

Furthermore, we define: 
\begin{itemize}
\item $\alpha_{i,k}(t)$: the number of timeslots within the first $t$ timelots, in which $k$ processing resource units were allocated at node $i$
\item $\beta_{i,k}^{(d,\phi,m)}(t)$: the number of packets of commodity $(d,\phi,m)$ processed by node $i$ during the $\alpha_{i,k}(t)$ timeslots in which $k$ processing resource units were allocated
\item $\alpha_{ij,k}(t)$: the number of timeslots within the first $t$ timeslots, in which $k$ transmission resource units were allocated at link $(i,j)$
\item $\beta_{ij,k}^{(d,\phi,m)}(t)$: the number of packets of commodity $(d,\phi,m)$ transmitted over link $(i,j)$ during the $\alpha_{ij,k}(t)$ timeslots in which $k$ transmission resource units were allocated
\end{itemize}

It then follows that
\begin{align}
\label{eq_process_average}
&\frac{F_{i,\text{pr}}^{(d,\phi,m)}(t)}{t}\le\frac{\alpha_{i,k}(t)}{t}\frac{\beta_{i,k}^{(d,\phi,m)}(t)r^{(\phi,m+1)}}{\alpha_{i,k}(t)C_{i,k}}\frac{C_{i,k}}{r^{(\phi,m+1)}}, \notag\\
&\qquad\qquad\qquad\qquad\qquad\qquad\qquad\quad\forall i,d,\phi,m<M_\phi,\\
\label{eq_transmission_average}
&\frac{F_{ij}^{(d,\phi,m)}(t)}{t}\le\frac{\alpha_{ij,k}(t)}{t}\frac{\beta_{ij,k}^{(d,\phi,m)}(t)}{\alpha_{ij,k}(t)C_{ij,k}}C_{ij,k},\quad \forall (i,j),d,\phi,m,
\end{align}
where we define $0/0=1$ in case of zero denominator terms.

Note that, for all $t$, we have
\begin{align}
\label{eq_process_bound}
&0\le \frac{\alpha_{i,k}(t)}{t}\le 1,\ 0\le \frac{\beta_{i,k}^{(d,\phi,m)}(t)r^{(\phi,m+1)}}{\alpha_{i,k}(t)C_{i,k}}\le 1,\\
\label{eq_transmission_bound}
&0\le \frac{\alpha_{ij,k}(t)}{t} \le 1,\ 0 \le \frac{\beta_{ij,k}^{(d,\phi,m)}(t)}{\alpha_{ij,k}(t)C_{ij,k}} \le 1.
\end{align}
In addition, let $\overline h$ represent the liminf of the average cost achieved by this policy:
\begin{equation}
\overline h \triangleq \mathop {\lim \inf }\limits_{t \to \infty } \frac{1}{t}\sum\nolimits_{\tau  = 0}^{t - 1} {h\left( \tau  \right)}.
\end{equation}
Then, due to Boltzano-Weierstrass theorem \cite{math_analysis} on a compact set, there exists an infinite subsequence $\{t_u\} \! \subseteq \! \{t\}$ such that
\begin{align}
\mathop {\lim }\limits_{t_u \to \infty } \frac{1}{t_u}\sum\nolimits_{\tau  = 0}^{t_u - 1} {h\left( \tau  \right)}=\overline h,
\end{align}
the left hand of \eqref{eq_process_average} and \eqref{eq_transmission_average} converge to $f_{i,\text{pr}}^{(d,\phi,m)}$ and $f_{ij}^{(d,\phi,m)}$:
\begin{align}
\label{eq_process_bound22}
&\mathop {\lim }\limits_{t_u \to \infty } \!\! \frac{F_{i,\text{pr}}^{(d,\phi,m)}\!(t_u)}{t_u}=f_{i,\text{pr}}^{(d,\phi,m)}\!, \mathop {\lim }\limits_{t_u \to \infty } \!\! \frac{F_{ij}^{(d,\phi,m)}\!(t_u)}{t_u}=f_{ij}^{(d,\phi,m)},
\end{align}
and the terms in \eqref{eq_process_bound} and \eqref{eq_transmission_bound} converge to $\alpha_{i,k}$, $\beta_{i,k}$, $\alpha_{ij,k}$, and $\beta_{ij,k}$:
\begin{align}
\label{eq_process_bound2}
&\mathop {\lim }\limits_{t_u \to \infty } \!\frac{\alpha_{i,k}(t_u)}{t_u}=\alpha_{i,k},\ \mathop {\lim }\limits_{t_u \to \infty } \!\frac{\beta_{i,k}^{(d,\phi,m)}(t_u)r^{(\phi,m+1)}}{\alpha_{i,k}(t_u)C_{i,k}}=\beta_{i,k},\\
\label{eq_transmission_bound2}
&\mathop {\lim }\limits_{t_u \to \infty }\frac{\alpha_{ij,k}(t_u)}{t_u}=\alpha_{ij,k},\ \mathop {\lim }\limits_{t_u \to \infty }\frac{\beta_{ij,k}^{(d,\phi,m)}(t_u)}{\alpha_{ij,k}(t)C_{ij,k}}=\beta_{ij,k}.
\end{align}
from which \eqref{alphas} and \eqref{betas} follow.

Plugging \eqref{eq_process_bound22}, \eqref{eq_process_bound2}, and \eqref{eq_transmission_bound2} respectively back into \eqref{eq_process_average} and \eqref{eq_transmission_average}, letting $t_u\rightarrow \infty$ yields
\begin{align}
\label{hao2}
& f_{i,\text{pr}}^{(d,\phi,m)} \leq \frac{1}{r^{(\phi,m+1)}}\alpha_{i,k}\beta_{i,k}^{(d,\phi,m)}C_{i,k},\\
& f_{ij}^{(d,\phi,m)} \leq \alpha_{ij,k}\beta_{ij,k}^{(d,\phi,m)}C_{ij,k},
\end{align}
from which \eqref{eq_thm1_processing_capacity_constraint} and \eqref{eq_thm1_flow_capacity_constraint} follow.

Furthermore, due to cloud network stability, we have
\begin{align}
\label{eq_lambda}
&\lim_{t \rightarrow \infty} \frac {\sum_{\tau=0}^t a_i^{(d,\phi,m)}(t)} {t} =  \lim_{t \rightarrow \infty}  \frac {X_i^{(d,\phi,m)}(t)} {t} =\lambda_i^{(d,\phi ,m)},\notag\\
&\qquad\qquad\qquad\qquad\qquad\qquad\qquad \text{w.p.1}, \ \forall i,d,\phi,m,\\
& \text{and} \notag\\
& \mathop {\lim }\limits_{{{  t}_u} \to \infty } \frac{F_{\text{pr},i}^{(d,\phi ,m)} \! (t_u)}{t_u} = \mathop {\lim }\limits_{{{  t}_u} \to \infty } \frac{\xi^{(\phi,m)} F_{i,\text{pr}}^{(d,\phi,m-1)}\!(t_u)}{t_u} \notag\\
&\qquad\qquad\qquad\qquad = {\xi ^{\left( {\phi ,m} \right)}}f_{i,\text{pr}}^{(d,\phi ,m - 1)} \notag\\
&\qquad\qquad\qquad\qquad \triangleq f_{\text{pr},i}^{(d,\phi ,m)}, \quad\, \text{w.p.1},\ \forall i,d,\phi,m,
\label{finiteD}
\end{align}
from which \eqref{eq_thm1_processing_conservation} follows.

Evaluating \eqref{eq_thm1_stabilityproof} in $\{t_u\}$, dividing by $t_u$, sending $t_u$ to $\infty$, and using \eqref{eq_process_bound22}, \eqref{eq_lambda}, and \eqref{finiteD},
Eq. \eqref{eq_thm1_stability} follows.

\pagenumbering{arabic}


Finally, from \eqref{obj}, and using the quantities defined at the beginning of this section,  we have
\begin{align}
&\frac{1}{t_u}\sum\nolimits_{\tau  = 0}^{t_u - 1} {h\left( \tau  \right)} \notag\\
& = \sum\limits_{i}\!\sum\limits_{k\in \mathcal K_i} \! \left[\!\frac{\alpha_{i,k}(t_u)w_{i,k}}{t_u} + \!\sum\limits_{(d,\phi,m)} \!\! \frac{r^{(\phi,m+1)}\beta_{i,k}^{(d,\phi,m)}\!(t_u)e_i}{t_u}\!\right] \notag\\
& + \sum\limits_{(i,j)}\!\sum\limits_{k\in \mathcal K_{ij}}\!\!\left[\frac{\alpha_{ij,k}(t_u)w_{ij,k}}{t_u} +\! \sum\limits_{(d,\phi,m)} \!\! \frac{\beta_{ij,k}^{(d,\phi,m)}\!(t_u) e_{ij}}{t_u}\right]\notag\\
& = \!\!\sum\limits_{i}\!\sum\limits_{k\in \mathcal K_i}\!\! \left[\Scale[1.04]{\!\frac{\alpha_{i,\!k}\!(t_u)w_{i,\!k}}{t_u} +  \frac{\alpha_{i,\!k}\!(t_u)}{t_u}} \!\!\!\!\sum\limits_{(d,\phi,m)} \!\!\!\!  \Scale[1.04]{\frac{r^{(\phi,m+1)}\beta_{i,k}^{(d,\phi,m)}\!(t_u) C_{i,k}{e_i}}{\alpha_{i,k}\!(t_u) C_{i,\!k}} } \!\right] \notag\\
& + \!\!\sum\limits_{(i,j)}\!\sum\limits_{k\in \mathcal K_{ij}}\!\!\left[ \Scale[1.05]{\frac{\alpha_{ij,k}\!(t_u)w_{ij,k}}{t_u} + \frac{\alpha_{ij,k}(t_u)}{t_u}} \!\!\!\sum\limits_{(d,\phi,m)}\!\!\!\!\Scale[1.04]{\frac{\beta_{ij,k}^{(d,\phi,m)}\!(t_u) C_{ij,k}e_{ij}}{\alpha_{ij,k}\!(t_u) C_{ij,k}} } \!\right] \!\! .
\end{align}
Letting $t_u\rightarrow \infty$, we obtain \eqref{eq:underline}. Finally, \eqref{eq_thm1_minimum_cost} follows from taking the minimum over all stabilizing policies.


\subsection{Proof of Sufficiency}

Given an input rate matrix $\bm\lambda \triangleq \{\lambda_i^{(d,\phi,m)} \}$, if there exits a constant $\kappa>0$ such that input rate $\{\lambda_i^{(d,\phi,m)}+\kappa \}$, together with probability values $\alpha_{ij,k}$, $\alpha_{i,k}$, $\beta_{ij,k}^{(d,\phi,m)}$, $\beta_{i,k}^{(d,\phi,m)}$, and flow variables $f_{ij}^{(d,\phi,m)}$,  $f_{i,\text{pr}}^{(d,\phi,m)}$, satisfy
(\ref{eq_thm1_stability})-(\ref{betas}),
we can construct a stationary randomized policy that uses these probabilities to make scheduling decisions, which yields the mean rates:
\begin{align}
\!\!\!\E \left\{ \mu_{i,\text{pr}}^{\left( {d,\phi ,m} \right)}\! (t)\right\}\! =\!f_{i,\text{pr}}^{(d,\phi,m)},\ \E \left\{ \mu_{ij}^{\left( {d,\phi ,m} \right)} \!(t)\right\} \!=\!  f_{ij}^{(d,\phi,m)}. \label{eq1}
\end{align}
Plugging \eqref{eq1} and $\{\lambda_i^{(d,\phi,m)}  + \kappa\}$ in  (\ref{eq_thm1_stability}), we have
\begin{eqnarray}
\label{suf1}
& \hspace{-0.3cm} \E \left \{\sum_{j \in \delta^{\!+\!}(i)}\! \mu_{ij}^{\left( {d,\phi ,m} \right)}\! (t) \!+\!   \mu_{i,\text{pr}}^{\left( {d,\phi ,m} \right)} \!(t)  \!-\!\sum_{j \in \delta^{\!-\!}(i)}\! \mu_{ji}^{\left( {d,\phi ,m} \right)} \!(t)\right.   \notag \\
&\hspace{-0.5cm}\left. -  \xi^{(d,\phi,m+1)}\mu_{i,\text{pr}}^{\left( {d,\phi ,m} \right)} (t) -a_i^{(d,\phi,m)(t)} \right\}
\geq \! \kappa .
\end{eqnarray}
By applying standard Lyapunov drift manipulations \cite{Neely_book}, we upper bound the Lyapunov drift $\Delta \left( {{\bf{Q}}\left( t \right)} \right)$ (see Sec. \ref{sec: LDP}) as
\begin{align}
\label{eq_stationary_LDP}
&\Scale[0.99]{\Delta\! \left( {{\bf{Q}}\!\left( t \right)} \right) \!\le \!N{B_0} \!+\! \sum\nolimits_{\left( {d,\phi ,m} \right),i} \!{Q_i^{\left( {d,\phi ,m} \right)}\!\!\left( t \right)\!\E\!\left\{\! {\sum\nolimits_{j\in \delta^{\!-\!}(i)}\! {\mu_{ji}^{(d,\phi ,m)}\!(t)} } \right.\! } }\nonumber\\
&\Scale[0.99]{+ \mu_{\text{pr},i}^{(d,\phi ,m)}\!(t)\! -\! \left. {\sum\nolimits_{j\in \delta^{\!+\!}(i)} {\mu_{ij}^{(d,\phi ,m)}\!(t)}  \!+\! \mu_{\text{pr},i}^{(d,\phi ,m)}\!(t) \!+\! a _i^{(d,\phi ,m)}\!(t)}\! \right\}}\nonumber\\
&\Scale[0.99]{ \le N{B_0} -\kappa\sum\nolimits_{\left( {d,\phi ,m} \right),i} {Q_i^{\left( {d,\phi ,m} \right)}\left( t \right)},}
\end{align}
where $B_0$ is a constant that depends on the system parameters.
With some additional manipulations, it follows from \eqref{eq_stationary_LDP} that the cloud network is strongly stable, \ie the total mean average
backlog is upper bounded. Therefore, $\{\lambda_i^{(d,\phi,m)} \}$ is interior to $\Lambda({\mathcal G},{ \Phi})$ (due to the existence of $\kappa$).

\section{Proof of Theorem \ref{thm: stability}}
\label{appendix_stability}

We prove Theorem \ref{thm: stability} for each DCNC algorithm by manipulating the linear term $Z(t)$ and the quadratic term $\Gamma(t)$ in the LDP upper bound expression given in \eqref{eq_lypunov_bound1}.

\subsection{DCNC-L}

We upper bound $\Gamma(t)$ in \eqref{eq_lypunov_bound1} as follows:
\begin{align}
\label{eq_L_second_moment_bound}
&\Gamma(t)\le \frac{1}{2}N\left[ {{{\left( {{\delta _{\max }}C_{\text{tr}}^{\max } + \left.C_{\text{pr}}^{\max }\right/r_{\min}} \right)}^2} + } \right.\nonumber\\
&\left. {{{\left( {{\delta _{\max }}C_{\text{tr}}^{\max } + \left.{\xi _{\max }}C_{\text{pr}}^{\max }\right/r_{\min} + {A_{\max }}} \right)}^2}} \right] \triangleq {NB_0}.
\end{align}
Plugging \eqref{eq_L_second_moment_bound}
into \eqref{eq_lypunov_bound1} yields
\begin{eqnarray}
\label{eq_lypunov_bound2}
&\hspace{-0.3cm}\Delta \left( {{\bf{Q}}\left( t \right)} \right) + V\mathbb{E}\left\{ {\left. {{h}(t)} \right|{\bf{Q}}\left( t \right)} \right\}\le \!N{B_{0}} \qquad\qquad\qquad\qquad\notag\\
&\hspace{-0.6cm}+ \mathbb{E}\left\{\! {\left. {V{h}(t)\!+\!Z(t)} \right|{\bf{Q}}\left( t \right)} \!\right\} \!+\!\!\sum\nolimits_{\left( {d,\phi ,m} \right),i} \!{\lambda _i^{\left( {d,\phi ,m} \right)}
Q_i^{\left( {d,\phi ,m} \right)}\!\!\left( t \right)}. 
\end{eqnarray}

Since ${\bm \lambda}\triangleq\{\lambda_i^{(d,\phi,m)}\}$ is interior to $\Lambda(\mathcal G,\Phi)$, there exists a positive number $\kappa$ such that ${\bm \lambda} + \kappa{\bf 1} \in \Lambda$.
According to \eqref{algor}, DCNC-L minimizes $Vh(t)+Z(t)$ among all policies subject to \eqref{cappr2}-\eqref{ys}.
We use $*$ to identify the stationary randomized policy that supports ${\bm \lambda} + \kappa{\bf 1}$ and achieves average cost ${\overline h}^*({\bm \lambda} + \kappa\bf 1)$, characterized by Theorem \ref{thm: network_capacity_region}.
The LDP function under DCNC-L can be further upper bounded as
\begin{align}
\label{eq_lypunov_bound3}
&\Scale[0.99]{\Delta \left( {{\bf{Q}}\left( t \right)} \right) + V\mathbb{E}\left\{ {\left. {{h}(t)} \right|{\bf{Q}}\left( t \right)} \right\} \le \!N{B_0}} \nonumber\\
&\Scale[0.99]{\ \ + \mathbb{E}\left\{\! {\left. {V{h^*}\!+\!Z^*\!(t)} \right|{\bf{Q}}\!\left( t \right)} \!\right\} \!+\!\!\! \sum\nolimits_{\left( {d,\phi ,m} \right),i} \!{\lambda _i^{\left( {d,\phi ,m} \right)}Q_i^{\left( {d,\phi ,m} \right)}\!\!\left( t \right)}}\nonumber
\end{align}
\begin{align}
&\Scale[0.99]{= N{B_0} + V{{\overline h}^*}\left( {\bm \lambda  + \kappa {\bf{1}}} \right)}\nonumber\\
&\Scale[0.99]{\ \ \ + \sum\nolimits_{\left( {d,\phi ,m} \right),i} {Q_i^{\left( {d,\phi ,m} \right)}\left( t \right)\left[ {\sum\nolimits_{j\in \delta^{\!-\!}(i)} {f_{ji}^{*(d,\phi ,m)}} } \right.\! +\! f_{\text{pr},i}^{*(d,\phi ,m)}} }\nonumber\\
&\Scale[0.99]{\ \ \ \qquad\qquad- \left. {\sum\nolimits_{j\in \delta^{\!+\!}(i)} {f_{ij}^{*(d,\phi ,m)}}  + f_{\text{pr},i}^{*(d,\phi ,m)} + \lambda _i^{(d,\phi ,m)}} \right]}\nonumber\\
&\Scale[0.99]{ \le N{B_0} + V{{\overline h}^*}\left( {\bm \lambda  + \kappa {\bf{1}}} \right)-\kappa\sum\nolimits_{\left( {d,\phi ,m} \right),i} {Q_i^{\left( {d,\phi ,m} \right)}\left( t \right)}.}
\end{align}
where the last inequality holds true due to \eqref{stab2}.

\vspace{-0.3cm}
\subsection{DCNC-Q}

We extract the quadratic terms $(\mu_{ij}^{(d,\phi,m)}(t))^2$ and $(\mu_{i,\text{pr}}^{(d,\phi,m)}(t))^2$ by decomposing $\Gamma(t)$ as follows:
\begin{align}
\label{eq_Phi}
\Gamma(t) = \Gamma _{\text{tr}}\left( t \right) + \Gamma _{\text{pr}}\left( t \right) + \Gamma'(t),
\end{align}
\vspace{-0.3cm}
where
\begin{align}
&\Gamma_{\text{tr}}\left( t \right) \triangleq \sum\nolimits_{(i,j)}  {\sum\nolimits_{(d,\phi,m)} {{{\left( {\mu _{ij}^{(d,\phi ,m)}\left( t \right)} \right)}^2}} }\nonumber\\
&\Gamma_{\text{pr}}\left( t \right) \triangleq \frac{1}{2} {\sum\nolimits_{(d,\phi ,m),i} \!{\left[ {{{\left( {\mu _{i,\text{pr}}^{(d,\phi ,m)}\left( t \right)} \right)}^2} \!+\! {{\left( {\mu _{\text{pr},i}^{(d,\phi ,m)}\left( t \right)} \right)}^2}} \right]} };\nonumber\\
&\Gamma'(t)\triangleq \sum\nolimits_{\left( {d,\phi ,m} \right),i} {\left\{ {\mu _{i,\text{pr}}^{(d,\phi ,m)}\left( t \right)\sum\nolimits_{j \in \delta \left( i \right)} {\mu _{ij}^{(d,\phi ,m)}\left( t \right)} } \right.}+\nonumber\\
& \! \sum\limits_{\scriptstyle j,v:j,v \in \delta \left( i \right), v \ne j}\!\!\! {\mu _{ij}^{(d,\phi ,m)}\!\!\left( t \right)\mu _{iv}^{(d,\phi ,m)}\!\!\left( t \right)} \! +\!  \mu _{\text{pr},i}^{(d,\phi ,m)}\!\!\left( t \right)\!\!\!\sum\limits_{j \in \delta \left( i \right)}\!\! {\mu _{ji}^{(d,\phi ,m)}\!\!\left( t \right)} \nonumber\\
&+\sum\nolimits_{\scriptstyle j,v: j,v \in \delta \left( i \right), v \ne j}\! {\mu _{ji}^{(d,\phi ,m)}\left( t \right)\mu _{vi}^{(d,\phi ,m)}\left( t \right)}+\frac{1}{2}\left(a_i^{(d,\phi,m)}\right)^2\nonumber\\
&+\left. {a_i^{(d,\phi ,m)}\left( t \right)\left[ {\sum\nolimits_{j \in \delta \left( i \right)} {\mu _{ji}^{(d,\phi ,m)}\left( t \right) + } \mu _{\text{pr},i}^{(d,\phi ,m)}\left( t \right)} \right]} \right\}.\nonumber
\end{align}

According to \eqref{algor_Q}, DCNC-Q minimizes the metric $\Gamma_{\text{tr}}\left( t \right) + \Gamma_{\text{pr}}\left( t \right) + Z(t) + Vh(t)$ among all policies subject to \eqref{cappr2}-\eqref{ys}.
Hence, the LDP function under DCNC-Q can be further upper bounded as
\begin{align}
\label{eq_lypunov_bound9}
&\Scale[0.99]{\Delta \left( {{\bf{Q}}\left( t \right)} \right) + V\mathbb{E}\left\{ {\left. {{h}(t)} \right|{\bf{Q}}\left( t \right)} \right\} } \nonumber\\
&\Scale[0.99]{\le \mathbb{E}\left\{\! {\left. {\Gamma'(t)\!+\!\Gamma^*_{\text{tr}} + \Gamma^*_{\text{pr}} } \right|{\bf{Q}}\!\left( t \right)} \!\right\} \!+\!V{\overline h^*}\!({\bm \lambda + \kappa {\bf 1}})}\nonumber\\
&\Scale[0.99]{ \quad+\E\left\{\left.\!Z^*\!(t)\right|{\bf Q}(t)\right\}\!+\! \sum\nolimits_{\left( {d,\phi ,m} \right),i} \!{\lambda _i^{\left( {d,\phi ,m} \right)}Q_i^{\left( {d,\phi ,m} \right)}\!\!\left( t \right)}}.
\end{align}
\vspace{-0.1cm}
On the other hand, note that
\vspace{-0.1cm}
\begin{align}
\label{eq_Q_second_moment_bound}
&\Gamma'(t)+\Gamma_{\text{tr}}^* + {\Gamma}_{\text{pr}}^*\nonumber\\
&\Scale[0.9]{\le N\left[ {\left.\left( {1 + {\xi _{\max }}} \right)C_{\text{pr}}^{\max }{\delta _{\max }}C_{\text{tr}}^{\max }\right/r_{\min} + \left( {{\delta _{\max }} - 1} \right){\delta _{\max }}
{{\left( {C_{\text{tr}}^{\max }} \right)}^2}} \right.}\nonumber\\
&\Scale[0.9]{\qquad+\left. {{A_{\max }}\left( {{\delta _{\max }}C_{\text{tr}}^{\max } + \left.{\xi _{\max }}C_{\text{pr}}^{\max }\right/r_{\min}} \right) + \frac{1}{2}{{\left( {{A_{\max }}} \right)}^2}} \right]}\nonumber\\
&\Scale[0.9]{\ \ \ \ + N{\delta _{\max }}{\left( {C_{\text{tr}}^{\max }} \right)^2} + \frac{1}{2(r_{\min})^2}N{\left( {C_{\text{pr}}^{\max }} \right)^2}\left[ {1 + {{\left( {{\xi _{\max }}} \right)}^2}} \right]}\nonumber\\
&\Scale[0.9]{= \frac{1}{2}N{\left[ {{{\left( {{\delta _{\max }}C_{\text{tr}}^{\max } + \left.C_{\text{pr}}^{\max }\right/r_{\min}} \right)}^2}} \right.}}\nonumber\\
&\Scale[0.9]{\qquad\ \ {\left.{+\! {{\left( {{\delta _{\max }}C_{\text{tr}}^{\max } \!+\! \left.\xi_{\max}C_{\text{pr}}^{\max }\right/r_{\min} \!+\! {A_{\max }}} \right)}^2}} \right]} \!=\! NB_0.}
\end{align}
Plugging \eqref{eq_Q_second_moment_bound} into \eqref{eq_lypunov_bound9} yield
\begin{eqnarray}
\label{eq_lypunov_bound5}
&\hspace{-4.4cm}\Delta \left( {{\bf{Q}}\left( t \right)} \right) + V\mathbb{E}\left\{ {\left. {{h}(t)} \right|{\bf{Q}}\left( t \right)} \right\}\notag\\ 
&\hspace{-1.4cm}\le NB_{0} + \overline h^*\left(\bm \lambda + \kappa \bf 1\right) -\kappa\sum\nolimits_{\left( {d,\phi ,m} \right),i} {Q_i^{\left( {d,\phi ,m} \right)}\left( t \right)}.
\end{eqnarray}

\subsection{EDCNC-L and EDCNC-Q}
Using \eqref{eq_abstract_backlog} in \eqref{eq_queueing_dynamic}, and  following standard LDP manipulations (see Ref. \cite{Neely_book2}), the LDP function can be upper bounded as follows:
\begin{align}
\label{eq_lp_bound_enhanced}
&\Delta \left( {{\bf{Q}}\left( t \right)} \right) + V\mathbb{E}\left\{ {\left. {{h}(t)} \right|{\bf{Q}}\left( t \right)} \right\}  \le  - \eta\Upsilon(t) \nonumber\\
&+NB_{0} + \overline h^*\left(\bm \lambda + \kappa \bf 1\right) -\kappa\sum\nolimits_{\left( {d,\phi ,m} \right),i} {\hat Q_i^{\left( {d,\phi ,m} \right)}\left( t \right)} ,
\end{align}
where
\begin{align}
\label{eq_terms_with_bias}
&\!\!\!\Scale[1]{\Upsilon\!(t) \!\triangleq\! \sum\nolimits_{(d,\phi ,m),i} \!{Y_i^{(d,\phi ,m)}\mathbb{E}\!\left\{\! {\sum\nolimits_{j \in \delta^- (i)} \!{\mu _{ji}^{(d,\phi ,m)}\!\!\left( t \right)}\! +\! a_i^{(d,\phi,m)}\!(t)} \right.} } \nonumber\\
&\Scale[1]{\left. {\left. { \!+ \mu _{\text{pr},i}^{(d,\phi ,m)}\!\!\left( t \right) \!-\! \sum\nolimits_{j \in \delta^+ (i)}\! {\mu _{ij}^{(d,\phi ,m)}\!\!\left( t \right)}  \!-\! \mu _{i,\text{pr}}^{(d,\phi ,m)}\!\!\left( t \right)} \right|\!{\bf{Q}}\left( t \right)}\! \right\}\!.}
\end{align}


Denote $Y_{\max}\triangleq\max\nolimits_{i,(d,\phi,m)}{Y_i^{(d,\phi,m)}}$, which satisfies $Y_{\max}\le \max\nolimits_{i,j}{\{H_{i,j}\}}\le N-1$. 
Then, following \eqref{eq_terms_with_bias}, we lower bound $\Upsilon(t)$ as
\begin{align}
\label{eq_bias_bound}
\Upsilon(t)\ge \Scale[0.97]{- N\left[\delta_{\max}C_{\text{tr}}^{\max} + \left.C_{\text{pr}}^{\max}\right/r_{\min}\right] \triangleq -NB_{\Upsilon}.}
\end{align}

Plugging \eqref{eq_bias_bound} into \eqref{eq_lp_bound_enhanced} and using $\hat Q_i^{(d,\phi,m)}(t)\ge Q_i^{(d,\phi,m)}(t)$ yields
\begin{align}
\label{eq_lypunov_bound7}
&\Delta \left( {{\bf{Q}}\left( t \right)} \right) + V\mathbb{E}\left\{ {\left. {{h}(t)} \right|{\bf{Q}}\left( t \right)} \right\}\notag\\
&\le NB_1 + \overline h^*\left(\bm \lambda + \kappa \bf 1\right) -\kappa\sum\nolimits_{\left( {d,\phi ,m} \right),i} {Q_i^{\left( {d,\phi ,m} \right)}\left( t \right)},
\end{align}
where $B_1\triangleq B_0 + \eta B_\Upsilon$.


\subsection{Network Stability and Average Cost Convergence with Probability 1}
We can use the theoretical result in \cite{Neely_prob_1} for the proof of network stability and average cost convergence with probability 1. Note that the following bounding conditions are satisfied in the cloud network system:
\begin{itemize}
\item The second moment $\mathbb{E}\{(h(t))^2\}$ is upper bounded by $(\sum\nolimits_{ij}{w_{ij,K_{ij}}} + \sum\nolimits_i{w_{i,K_i}})^2$ and therefore satisfies
\begin{equation}
\label{eq_prob1_bound1}
\sum\nolimits_{\tau  = 0}^\infty  {\left.{{\mathbb{E}\left\{ {{{\left( {h\left( \tau  \right)} \right)}^2}} \right\}}}\right/{\tau }}  < \infty.
\end{equation}

\item $\mathbb{E}\{\left.h(t)\right|{\bf Q}(t)\}$ is lower bounded as
\label{eq_prob1_bound2}
\begin{equation}
\mathbb{E}\{\left.h(t)\right|{\bf Q}(t)\}\ge 0.
\end{equation}

\item For all $i$, $(d,\phi,m)$, and $t$, the conditional fourth moment of backlog dynamics satisfies
\begin{align}
\label{eq_prob1_bound3}
&\!\!\!\!\mathbb{E}\left\{ {\left. {{{\left[ {Q_i^{\left( {d,\phi ,m} \right)}\left( {t + 1} \right) - Q_i^{\left( {d,\phi ,m} \right)}\left( t \right)} \right]}^4}} \right|{\bf{Q}}(t)} \right\}\nonumber\\
&\!\!\!\!\!\!\le\!\! \left(\delta_{\max}C_{\text{tr}}^{\max} + \left.\xi_{\max}C_{\text{pr}}^{\max}\right/r_{\min}+ A_{\max}\right)^4\!<\!\infty.
\end{align}
\end{itemize}
With \eqref{eq_prob1_bound1}-\eqref{eq_prob1_bound3}, based on the derivations in \cite{Neely_prob_1}, Eq. \eqref{eq_lypunov_bound3}, \eqref{eq_lypunov_bound5}, and \eqref{eq_lypunov_bound7} lead to network stability \eqref{eq_que} and average cost \eqref{eq_average_cost} convergence with probability 1 under DCNC-L, DCNC-Q, EDCNC-L(Q), respectively.

\section{Proof of Theorem \ref{thm: convergence_rate}}
\label{appendix_converge_rate}
Let's first prove Eq. \eqref{eq_cost_converge_rate}. To this end  denote $\overline h (t)\triangleq \frac{1}{t}\sum\nolimits_{\tau =0}^{t-1} \E\{h(\tau)\}$. Then, under the DCNC policy and after some algebraic manipulations similar to the ones used for \eqref{eq_lypunov_bound3}, we upper bound the LDP function as follows:
\begin{equation}
\label{eq_ldp_upper_bound1}
\Delta({\bf Q}(t)) + V\E\{\left.h(t)\right|{\bf Q}(t)\} \le NB + V\overline h^*(\bm \lambda),
\end{equation}
where $\overline h^*(\bm \lambda)$ is the minimum average cost given $\bm \lambda$. Taking the expectation over ${\bf Q}(t)$ on both sides of \eqref{eq_ldp_upper_bound1} and summing over $\tau = 0,\cdots,t-1$ further yields
\begin{equation}
\label{eq_sum_LDP}
\frac{1}{{2t}}\!\left[\! {\mathbb{E}\left\{ {{{\left\| {{\bf{Q}}\left( t \right)} \right\|}^2}} \right\} \!-\! \mathbb{E}\left\{ {{{\left\| {{\bf{Q}}\left( 0 \right)} \right\|}^2}} \right\}} \!\right] \!\le\! NB \!+\! V\left[ {{{\overline h}^*}\!({\bm \lambda} ) \!-\! \overline h\!\left( t \right)} \right].
\end{equation}
Then it follows that, by setting $V=1/\epsilon$ and for all $t\ge 1$,
\begin{align}
\label{eq_average_cost2}
\overline h(t) - {{\overline h}^*}({\bm \lambda} ) &\le \frac{{NB}}{V} + \frac{1}{{2Vt}}{\mathbb{E}\left\{ {{{\left\| {{\bf{Q}}\left( 0 \right)} \right\|}^2}} \right\}}\notag\\
&\le \left[NB + \frac{1}{2}\E\left\{\left\|{\bf Q}(0)\right\|^2\right\}\right]\epsilon,
\end{align}
which proves \eqref{eq_cost_converge_rate}.

In order to prove \eqref{eq_flow_conserv_converge_rate}, we  first
introduce  the following quantities for an arbitrary policy:
\begin{itemize}
\item ${\bf y}(t)$: the vector of elements $y_i(t)$ and $y_{ij}(t)$;
\item $\tilde {\bm \mu}(t)$: the vector of actual flow rate elements $\tilde \mu_{i,\text{pr}}^{(d,\phi,m)}(t)$, $\tilde \mu_{\text{pr},i}^{(d,\phi,m)}(t)$, and $\tilde \mu_{ij}^{(d,\phi,m)}(t)$;
\item $\tilde {\bf x}(t)$: the vector $[{\bf y}(t); \tilde {\bm \mu}(t)]$;
\item $\Delta {\bf f}(t)$: the vector of elements $\Delta f_{i}^{(d,\phi,m)}(t)$  as in \eqref{eq_queueing_dynamic2}.
\end{itemize}
Summing both sides of \eqref{eq_queueing_dynamic2} over $\tau = 0,1,\cdots,t-1$ and then dividing them by $t$, for all $i,d,\phi,m,t\ge 1$, yield,
\begin{align}
\label{eq_sum_net_rate}
&\overline{\Delta {\bf f}}(t)\triangleq\frac{1}{t}\sum\nolimits_{\tau  = 0}^{t - 1}  \E\{\Delta {\bf f}(\tau)\}  \!=\! \frac{1}{t}\E\left\{{\bf Q}(t) \!-\! {\bf Q}\left( 0 \right)\right\}.
\end{align}
\begin{lem}
\label{lem: lemma}
If $\bm \lambda$ is interior to $\Lambda(\mathcal G,\Phi)$, there exits a constant vector $\bm \rho$ such that
\begin{equation}
\label{eq_supporting_hyper_plane}
\overline h^*(\bm \lambda) \!-\! \overline h(t) \!\le\! {\bm \rho}^\dag \overline{\Delta {\bf f}}(t).
\end{equation}
\end{lem}
\begin{proof}
The proof is given in Appendix \ref{appendix_proof_lemma} 
\end{proof}

Plugging \eqref{eq_sum_net_rate} into \eqref{eq_supporting_hyper_plane} yields
\begin{align}
\label{eq_farkas_result}
&\overline h^*(\bm \lambda) - \overline h(t) \le  \frac{1}{t}{\bm \rho}^\dag \mathbb{E}\left\{{\bf Q}(t)-{\bf Q}(0)\right\}\nonumber\\
&\le \frac{1}{t}\left\| {\bm \rho}  \right\|\cdot\left( {\left\| \mathbb{E}\{{\bf Q}\left( t \right)\} \right\| + \left\| \mathbb{E}\{{\bf Q}\left( 0 \right)\} \right\|} \right),
\end{align}
Under the DCNC policy, by further plugging \eqref{eq_farkas_result}  into the right hand side of \eqref{eq_sum_LDP}, we have
\begin{align}
\label{eq_sum_LDP2}
&\frac{1}{{2t}}\!\left(\! {\mathbb{E}\left\{ {{{\left\| {{\bf{Q}}\left( t \right)} \right\|}^2}} \right\} \!-\! \mathbb{E}\left\{ {{{\left\| {{\bf{Q}}\left( 0 \right)} \right\|}^2}} \right\}} \!\right)\nonumber\\
&\le NB + V { \frac{\left\| {\bm \rho}  \right\|}{t}\left( {\left\| {\mathbb{E}\left\{ {{\bf Q}\left( t \right)} \right\}} \right\| + \left\| {\mathbb{E}\left\{ {{\bf Q}\left( 0 \right)} \right\}} \right\|} \right)}.
\end{align}
By using the fact ${\mathbb{E}\{ {{{\left\| {{\bf{Q}}\left( t \right)} \right\|}^2}} \}} \ge \left\| \mathbb{E}\{{\bf{Q}}\left( t \right)\} \right\|^2$ in \eqref{eq_sum_LDP2}, it follows that
\begin{align}
\label{eq_second_order_inequality}
&{\left\| \mathbb{E}\{{\bf{Q}}\left( t \right)\} \right\|^2} - 2V\left\| {\bm \rho}  \right\|\cdot\left\| \mathbb{E}\{{\bf{Q}}\left( t \right)\} \right\| - 2NBt \notag\\
& \quad- {\mathbb{E}\{ {{{\left\| {{\bf{Q}}\left( 0 \right)} \right\|}^2}} \}} - 2V\left\| {\bm \rho}  \right\|\cdot\left\| \mathbb{E}\{{\bf{Q}}\left( 0 \right)\} \right\|\le 0.
\end{align}
The largest value of $\|\E\{{\bf Q}(t)\}\|$ that satisfies \eqref{eq_second_order_inequality} is given by
\begin{align}
\label{eq_queue_upper_bound_second_order}
&\left\| \mathbb{E}\{{\bf{Q}}\left( t \right)\} \right\| \le V\left\| {\bm \rho}  \right\| + \notag\\
&\sqrt {{V^2}{{\left\| {\bm \rho}  \right\|}^2} + 2NBt + {\mathbb{E}\{ {{{\left\| {{\bf{Q}}\left( 0 \right)} \right\|}^2}} \}} + 2V\left\| {\bm \rho}  \right\|\cdot\left\| \mathbb{E}\{{\bf{Q}}\left( 0 \right)\} \right\|}\notag\\
&\le\!V\left\| {\bm \rho}  \right\| \!+\! \sqrt{\left(V\left\| {\bm \rho}  \right\| \!+\! \E\{\left\|{\bf Q}(0)\right\|\}\right)^2 \!+\! 2NBt \!+\! {\text{var}}\{\left\|{\bf Q}(0)\right\|\}}.
\end{align}

Finally, by setting $V=1/\epsilon$ and $t=1/\epsilon^2$, we plug \eqref{eq_queue_upper_bound_second_order} back into the right hand side of \eqref{eq_sum_net_rate} and obtain
\begin{align}
&\frac{1}{t}\sum\nolimits_{\tau  = 0}^{t - 1}  {\mathbb{E}\left\{\Delta f_i^{(d,\phi ,m)}\left( \tau  \right)\right\}}  \le \frac{1}{t}\left(\left\|\mathbb{E}\{{\bf Q}(t)\}\right\| + \left\|\mathbb{E}\{{\bf Q}(0)\}\right\|\right)\nonumber\\
&\le \frac{V\left\| {\bm \rho}  \right\|}{t} + \frac{\left\|\mathbb{E}\{{\bf Q}(0)\}\right\|}{t}\nonumber\\
&\qquad +\frac{1}{t}\sqrt{\left(V\left\| {\bm \rho}  \right\| \!+\! \E\{\left\|{\bf Q}(0)\right\|\}\right)^2 \!+\! 2NBt \!+\! {\text{var}}\{\left\|{\bf Q}(0)\right\|\}}\notag\\
&\le\left( 2\mathbb{E}\{\left\|{\bf Q}(0)\right\|\}\!+\!\sqrt{{\text{var}}\{\left\|{\bf Q}(0)\right\|\}}\right)\!\epsilon^2 \!+\! \left(\sqrt{2NB}\!+\!2\left\| {\bm \rho}  \right\|\right)\!\epsilon,
\end{align}
which proves \eqref{eq_flow_conserv_converge_rate}.


\section{Proof of Lemma \ref{lem: lemma}}
\label{appendix_proof_lemma}

Given an arbitrary policy, define
\begin{itemize}
\item $\bm \mu(t)$: the vector of assigned flow rate elements $ \mu_{i,\text{pr}}^{(d,\phi,m)}(t)$, $ \mu_{\text{pr},i}^{(d,\phi,m)}(t)$, and $ \mu_{ij}^{(d,\phi,m)}(t)$;
\item ${\bf x}(t)$: the vector $[{\bf y}(t);{\bf x}(t)]$.
\end{itemize}
With a little bit abuse of notation, denote $h({\bf x}(t))\triangleq h(t)$; $ \Delta {\bf f}({\tilde {\bf x}(t)})\triangleq\Delta {\bf f}(t)$. In addition, let  ${\mathcal X}$ represent the set of all possible vectors ${\bf x}(t)$ that satisfy the
constraints \eqref{chain2}-\eqref{ys}. Note that $\tilde {\bf x}(t)$ also belongs to $\mathcal X$. Furthermore, let  $\overline {\mathcal X}$ represent the convex hull of ${\mathcal X}$.
Then, for all vectors ${\bf x}\in \overline {\mathcal X}$, the following conditions are satisfied:
\begin{enumerate}
\item $h({\bf x})-\overline h^*({\bm \lambda})$ and $\Delta {\bf f}({\bf x})$ are convex for all ${\bf x}\in \overline {\mathcal X}$;
\item $h({\bf x})-\overline h^*({\bm \lambda})\ge 0$ for all ${\bf x}\in \overline {\mathcal X}$ with $\Delta {\bf f}({\bf x})\preceq 0$;
\item $\exists\hat {\bf x}\in \overline{\mathcal X}$ with $\Delta {\bf f}(\hat{\bf x})\prec 0$, given $\bm \lambda$ interior to $\Lambda({\mathcal G},\Phi)$.
\end{enumerate}
Item 2) above results immediately from Theorem \ref{thm: network_capacity_region}, where any ${\bf x}\in \overline {\mathcal X}$ with $\Delta {\bf f}({\bf x})\preceq 0$ can be treated as the $\E\{{\bf x}(t)\}$ under a stabilizing stationary randomized policy.
Hence, according to Farkas' Lemma \cite{Bertsekas_convex_opt},
there exists a constant vector $\bm \rho\succeq 0$ such that
\begin{equation}
\label{eq_seperating_hyperplane2}
h({\bf x})-\overline h^* ({\bm \lambda})+ {\bm \rho}^\dag \Delta{\bf f}({\bf x}) \ge 0,\quad \forall {\bf x}\in \overline {\mathcal X}.
\end{equation}

Evaluating \eqref{eq_seperating_hyperplane2} in $\tilde {\bf x}(\tau)$ with $\tau=0, \ldots, t-1$, we have
\begin{equation}
\label{eq_seperating_hyperplane3}
\frac{1}{t}\sum\nolimits_{\tau=0}^{t-1}h(\tilde {\bf x}(\tau))-\overline h^* ({\bm \lambda})+ \frac{1}{t}{\bm \rho}^\dag \sum\nolimits_{\tau=0}^{t-1} \Delta{\bf f}(\tilde {\bf x}(\tau)) \!\ge\! 0,
\end{equation}
from which it follows that
\begin{align}
\!{\bm \rho}^\dag \overline{\Delta {\bf f}}(t) \!&=\! \frac{1}{t} {\bm \rho}^\dag \sum\limits_{\tau=0}^{t-1} \E\{\Delta{\bf f}(\tilde {\bf x}(t))\} \!\ge\! \overline h^* ({\bm \lambda}) \!-\! \frac{1}{t} \sum\limits_{\tau=0}^{t-1}\E\{h(\tilde {\bf x}(\tau))\}\notag\\
&\buildrel (a) \over \ge\! \overline h^* ({\bm \lambda}) \!-\! \frac{1}{t} \sum\limits_{\tau=0}^{t-1}\E\{h( {\bf x}(\tau))\} \!=\! \overline h^* ({\bm \lambda}) \!-\! \overline h(t),
\end{align}
where the inequality $(a)$ is due to $h(\tilde {\bf x}(t))\le h( {\bf x}(t))$ that results from the fact $\tilde {\bm \mu}(t)\preceq {\bm \mu}(t)$.

\IEEEpeerreviewmaketitle


\begin{thebibliography}{1}


\bibitem{Infocom_workshop} H. Feng, J. Llorca, A. M. Tulino, and A. F. Molisch, ``Dynamic Network Service Optimization in Distributed Cloud Networks,''
{\em IEEE INFOCOM SWFAN Workshop}, April 2016.

\bibitem{icc_2016} H. Feng, J. Llorca, A. M. Tulino, and A. F. Molisch,
``Optimal Dynamic Cloud Network Control,''
\emph{IEEE ICC}, 2016.

\bibitem{nfv_sdn} Bell Labs Strategic White Paper,
``The Programmable Cloud Network - A Primer on SDN and NFV,''
June 2013.

\bibitem{fxbook} Marcus Weldon,
``The Future X Network,''
\emph{CRC Press}, October 2015.

\bibitem{vnf} L. Lewin-Eytan, J. Naor, R. Cohen, and D. Raz,
``Near Optimal Placement of Virtual Network Functions,''
\emph{IEEE INFOCOM}, 2015.

\bibitem{csdp} M. Barcelo, J. Llorca, A. M. Tulino, and N. Raman,
``The Cloud Servide Distribution Problem in Distributed Cloud Networks,''
\emph{IEEE ICC}, 2015.

\bibitem{Neely_book} L. Georgiadis, M. J. Neely, and L. Tassiulas, ``Resource allocation and cross-layer control in wireless networks,''
{\em Now Publishers Inc.}, {2006}.

\bibitem{Neely_energy_control} M. J. Neely, ``Energy optimal control for time-varying wireless networks,''
{\em IEEE Transactions on Information Theory}, vol. {52}, pp. {2915--2934}, July, {2006}.

\bibitem{Neely_book2} M. J. Neely, ``Stochastic network optimization with application to communication and queueing systems,''
{\em Synthesis Lectures on Communication Networks},  {Morgan \& Claypool}, vol. {3}, pp. {1--211}, {2010}.

\bibitem{backpressure} L. Tassiulas, and A. Ephremides, ``Stability properties of constrained queueing systems and scheduling
policies for maximum throughput in multihop radio networks,''
{\em IEEE Transactions on Automatic
Control}, vol. {37}, no. 12, pp. {1936-1948}, Dec., {1992}.

\bibitem{Neely_convergence} M. J. Neely, ``A simple convergence time analysis of drift-plus-penalty for stochastic optimization and convex programs,''
\emph{arXiv preprint arXiv:1412.0791}, {2014}.

\bibitem{Neely_DIVBAR} M. J. Neely, ``Optimal Backpressure Routing for Wireless Networks with Multi-Receiver Diversity'',
{\em Ad Hoc Networks}, vol. {7}, pp. {862--881}, {2009}.

\bibitem{Neely_2005} M. J. Neely, E. Modiano and C. E. Rohrs, "Dynamic power allocation and routing for time-varying wireless networks",
{\em Selected Areas in Communications, IEEE Journal on}, vol. {23}, pp. {89-103}, Jan., {2005}.

\bibitem{Sucha_second_moment} S. Supittayapornpong and M. J. Neely, ``Quality of information maximization for wireless networks via a fully separable quadratic policy,''
{\em IEEE Transactions on Information Theory}, vol. {52}, pp. {2915--2934}, July, {2006}.


\bibitem{fog} M. Chiang and T. Zhang, ``Fog and IoT: An Overview of Research Opportunities'', \emph{IEEE Internet of Things Journal}, vol. {3}, no. {6}, pp. {854-864}, Dec. {2016}.

\bibitem{cloudlet} M. Satyanarayanan, P. Bahl, R. Caceres and N. Davies, ``The Case for VM-Based Cloudlets in Mobile Computing'', \emph{IEEE Pervasive Computing}, vol. {8}, no. {4}, pp. {14-23}, Oct.-Dec. {2009}.


\bibitem{iot}
S.~Nastic, S.~Sehic, D.-H. Le, H.-L. Truong, and S.~Dustdar,
``Provisioning software-defined IoT cloud systems,''
\emph{Future Internet of Things and Cloud (FiCloud)}, pp. 288--295, 2014.

\bibitem{iotcloud_jsac16}
M. Barcelo, A. Correa, J. Llorca, A. M. Tulino, J. L. Vicario and A. Morell,
``IoT-Cloud Service Optimization in Next Generation Smart Environments,"
\emph{IEEE Journal on Selected Areas in Communications}, vol. 34, no. 12, pp. 4077-4090, 2016.


\bibitem{Algorithm_Kleignberg} J. Kleinberg, and E. Tardos, ``Algorithm design,''
\emph{Pearson Education India}, {2006}.

\bibitem{optimization_boyd} S. Boyd, and L. Vandenberghe, ``Convex optimization,''
\emph{Cambridge university press}, {2004}.




\bibitem{Neely_prob_1} M. J. Neely, ``Queue Stability and Probability 1 Convergence via Lyapunov Optimization,''
\emph{arXiv preprint arXiv:1008.3519}, {2010}.

\bibitem{math_analysis} W. Rudin, ``Principles of mathematical analysis,''
\emph{New York: McGraw-hill}, vol. 3, {1964}.

\bibitem{Bertsekas_convex_opt} D. P. Bertsekas, ``Convex optimization theory,''
\emph{Belmont: Athena Scientific}, {2009}.


\end{thebibliography}
\end{document}